\documentclass[manyauthors]{fundam}

\publyear{22}
\papernumber{2102}
\volume{185}
\issue{1}

\usepackage{cleveref}
\usepackage{url}
\usepackage[ruled,lined]{algorithm2e}
\usepackage{graphicx}
\usepackage{mathrsfs}
\usepackage{amsmath,amssymb,epsfig}
\usepackage{mathpartir}
\usepackage{tikz,url}
\usepackage{tikz-cd}
\usepackage{sty/tikzit}
\usepackage{tikz-3dplot}

\tikzstyle{white vertex}=[fill=white, draw=black, shape=circle, thick, inner sep=1pt, minimum size=6pt]
\tikzstyle{green vertex}=[fill=green, draw=black, shape=circle, thick, inner sep=1pt, minimum size=6pt]
\tikzstyle{cyan vertex}=[fill=cyan, draw=black, shape=circle, thick, inner sep=1pt, minimum size=6pt]
\tikzstyle{magenta vertex}=[fill={magenta!60}, draw=black, shape=circle, thick, inner sep=1pt, minimum size=6pt]
\tikzstyle{empty node}=[fill=white, shape=circle, inner sep=1pt]
\tikzstyle{small font vertex}=[fill=none, draw=none, shape=circle, font={{\scriptsize}}]
\tikzstyle{new style 0}=[fill=none, draw=none, shape=circle]

\tikzstyle{filled path}=[-, fill={blue!60}, thick, fill opacity=0.15]
\tikzstyle{thick edge}=[-, thick]
\tikzstyle{dashed path}=[-, dashed]
\tikzstyle{arrow path}=[->]
\tikzstyle{arrow}=[->, >=stealth, thick]

\usetikzlibrary{arrows,automata,decorations,decorations.pathmorphing,fit}
\tdplotsetmaincoords{70}{115}
\usetikzlibrary{matrix}
\tikzset{
  curarrow/.style={
  rounded corners=8pt,
  execute at begin to={every node/.style={fill=red}},
    to path={-- ([xshift=-50pt]\tikztostart.center)
    |- (#1) node[fill=white] {$\scriptstyle \partial_*$}
    -| ([xshift=50pt]\tikztotarget.center)
    -- (\tikztotarget)}
    }
}

\usepackage{macros}

\usepackage{todonotes} 
\makeatletter
\newcommand\listofTODO{\section*{Remaining TODO}\@starttoc{tdo}}
\makeatother

\newcommand{\RK}[1]{\textcolor{blue}{#1}}

\begin{document}

\title{A categorical and logical framework for iterated protocols}

\address{eric.goubault@lix.polytechnique.fr}

\author{Eric Goubault \\ LIX \\ CNRS, École polytechnique, IP Paris \\ 91120 Palaiseau, France 
  \and Bernardo {Hummes Flores} \\ LIX \\ CNRS, École polytechnique, IP Paris \\ 91120 Palaiseau, France 
  \and Roman Kniazev \\ LaBRI \\ CNRS, Université de Bordeaux \\ 33405 Talence, France 
  \and J\'er\'emy Ledent \\ IRIF \\ CNRS, Université Paris Cité \\ 75013 Paris, France 
  \and Sergio Rajsbaum \\ Instituto de Matem\'aticas \\ Universidad Nacional Aut\'onoma de M\'exico \\ 04510 Mexico City, Mexico }

\maketitle


\runninghead{E. Goubault, B. {Hummes Flores}, R. Kniazev, J. Ledent, S. Rajsbaum}{\\A categorical and logical framework for iterated protocols}

\begin{abstract}
In this article, we show that the now classical protocol complex approach to distributed task solvability of Herlihy et al. \cite{herlihy} can be understood in standard categorical terms. First, protocol complexes are functors, from chromatic (semi-) simplicial sets to chromatic simplicial sets, that naturally give rise to algebras. These algebras describe the next state operator for the corresponding distributed systems. This is constructed for semi-synchronous distributed systems with general patterns of communication for which we show that these functors are always Yoneda extensions of simpler functors, implying a number of interesting properties. Furthermore, for these protocol complex functors, we prove the existence of a free algebra on any initial chromatic simplicial complex, modeling iterated protocol complexes. Under this categorical formalization, protocol complexes are seen as transition systems, where states are structured as chromatic simplicial sets. We exploit the epistemic interpretation of chromatic simplicial sets \cite{GoubaultLR21simplicial} 
and the underlying transition system (or algebra) structure to introduce a temporal-epistemic logic and its semantics on all free algebras on chromatic simplicial sets. We end up by giving hints on how to extend this framework to more general dynamic network graphs and state-dependent protocols, and give example in fault-tolerant distributed systems and mobile robotics. 
\end{abstract}

\begin{keywords}
Distributed systems, protocol complex, algebras and co-algebras, free algebras, temporal logics, epistemic logics.
\end{keywords}

\paragraph{\bf ACM subject classification:} 
CCS$\rightarrow$Theory of computation$\rightarrow$Models of computation$\rightarrow$Concurren\-cy$\rightarrow$Distributed computing models; 
CCS$\rightarrow$Theory of computation$\rightarrow$Logic$\rightarrow$Modal and temporal logics; CCS$\rightarrow$Theory of computation$\rightarrow$Logic$\rightarrow$Logic and verification


\section{Introduction} 






One of the great successes in formal methods for programming languages is the use of temporal logics \cite{LTL} for specifying and verifying sequential systems. This has been originally based on linear-time temporal logics (LTL), in which the progress of execution is identified with a step number, an integral number. The semantics of LTL is easily given in full generality on transition systems (Kripke structures of some particular sort), where next step and until operators are easily interpretable. This makes it possible to relate LTL properties with sequential programs, whose semantics are classically given in terms of transition systems. 

In distributed systems, many specification and verification frameworks have been proposed, for different application areas, e.g. linearizability~\cite{linearizability} for concurrent data structures, task specifications~\cite{MoranW87,BiranMZ90} for fault-tolerant distributed protocols (e.g. consensus), partial unifications such as the ones of \cite{unifying,opodis18}, etc. but seem to elude the world of logics and model-checking, at the notable exception of \cite{FHMVbook} using epistemic logics and \cite{Knight} for concurrent systems.

In this paper, we show how to build on the success of epistemic logics for multiagent systems and of temporal logics for concurrent systems, so as to get a temporal-epistemic logics suited to round-based wait-free distributed systems. The main observation of the paper, that makes this possible, is to see that the protocol complex approach to distributed computability \cite{herlihy} can be naturally recast as particular functors $\Fun{F}$ from chromatic semi-simplicial sets to chromatic simplicial sets. These chromatic semi-simplicial sets represent the set of potential global states of the distributed system at some round of communication, and this ``protocol complex functor'' $\Fun{F}$ produces from this set of global states, the set of potential next states, still organized within a chromatic semi-simplicial set. Under some mild assumptions, we show that this functor $\Fun{F}$ is naturally a Yoneda extension of a simpler functor, hence exhibits numerous interesting categorical properties. Moreover, we show that this naturally produces a $\Fun{F}$-algebra, and that the interesting categorical properties of $\Fun{F}$ we were mentioning imply that there is always free $\Fun{F}$-algebras. This allows for formalizing nicely iterated protocol complexes, and for deriving the semantics of a temporal-epistemic logics that formalizes the temporal evolution of the knowledge within such round-based distributed systems. This is based on previous work \cite{infcomp,boletin,LICS2023} showing that chromatic semi-simplicial sets provide natural semantics to epistemic logics.

\paragraph{Contents}


We begin by giving the necessary background material in Section~\ref{sec:algebras}, to set up the scene on (categorical) algebras and co-algebras, and in Section~\ref{sec:roundbased}, with distributed protocols and the use of simplicial complexes and sets as formalizations of their states.

We then show in Section~\ref{sec:protocolfunctor}, informally at first, that these protocol complexes naturally give rise to functors, and even very particular functors in the case of oblivious protocols that are Yoneda extensions of simple functors defined on sets of participating processes.

The protocol complex functor encodes the set (or simplicial set) of states that are reachable after one step of communication, from some given set (or, more precisely again, simplicial set) of initial states. The way the next states are linked to any particular initial state is through a simplicial set morphism from the protocol complex to the initial complex, meaning that the transition to next states can be encoded in algebras, similarly to the case of more classical transition systems. This is developed in Section~\ref{sec:protfunctor}, where we also mention a co-algebraic view, that comes for free due to the fact our protocol complex functors always admit a right adjoint.

We then prove in Section~\ref{sec:iterated} that our protocol complex functors always admit free algebras, which allows for naturally formalizing iterated protocols.

We introduce in Section~\ref{sec:temporalepistemic} a temporal-epistemic logics which is particularly well suited to an interpretation in free algebras. It naturally extends DEL as used for particular protocols \cite{gandalf} and give a few examples of potential applications to distributed systems and mobile robotics.

Finally, in Section~\ref{sec:beyond}, we give hints about how to generalize this work to non-oblivious protocols. This includes, in particular, general dynamic networks \cite{dynamicnetworks} and mobile robots.











\section{Algebras (and co-algebras) and sequential systems, a reminder} 
\label{sec:algebras}
Algebras, in the category-theoretic sense \cite{SMLCategories}, appear in a number of situations. For instance, any deterministic transition system can be seen as an algebra over a certain endofunctor on the category of sets. Given a set $Ac$ of ``actions''', consider the endofunctor $\Fun{F}$ on sets $X \mapsto X \times Ac$. Then a function $A: X \times Ac \to X$ gives rise to an automaton: indeed, given a pair $(x, a)$, where $x$ is interpreted as a state and $a$ as an action, the function $A$ assigns to this pair another state $x'$, the end-point of the transition. This construction is an instance of an $\Fun{F}$-algebra that we define below:

\begin{definition}[\(\Fun{F}\)-algebra]
Let $\Fun{F}$ be an endofunctor on a category $\mathcal{C}$, $\Fun{F}: \mathcal{C} \rightarrow \mathcal{C}$. An $\Fun{F}$-algebra is a pair $(C,a: \Fun{F}(C)\rightarrow C)$ where $C \in \mathcal{C}$ and $a$ is a morphism in $\mathcal{C}$, called the \emph{algebra structure} of $C$.

Let $(C,a)$ and $(C',a')$ be two $\Fun{F}$-algebras. Then $f$ is a morphism of $\Fun{F}$-algebras from $(C,a)$ to $(C',a')$ if and only $f$ is a morphism from $C$ to $C'$ in $\mathcal{C}$ and the following diagram commutes: 
\begin{center}
\begin{tikzcd}[row sep=large]
\Fun{F}(C) \arrow[d, "\Fun{F}(f)"] \arrow[r, "a"]  &  C \arrow[d, "f"]     \\
F(C')                 \arrow[r, "a'"] &  C'                    \\
\end{tikzcd}
\end{center}
The category of $\Fun{F}$-algebras is denoted by $Alg(\Fun{F})$. 
\end{definition}

There is a dual notion of a co-algebra, which is an algebra in the opposite category, that is, the direction of the morphisms is reversed.
We will touch upon the subject only in Section \ref{sec:coalgebras}, in which we mention the cases where we have the same ``object'' being both an $\Fun{F}$-algebra and a $\Fun{G}$-co-algebra, for some functor $\Fun{G}$ related to $\Fun{F}$ (this will be the case for our models of distributed systems). 

In their seminal work \cite{Varieties}, Adámek and Porst argue that, to model sequential systems via $\Fun{F}$-algebras, it is desirable for the endofunctor $\Fun{F}$ to admit free algebras. That is, for every object $I$, there should exist a free $\Fun{F}$-algebra generated by $I$:

\begin{definition}[Free \(\Fun{F}\)-algebra]
\label{def:free}
Given an endofunctor $\Fun{F}: \mathcal{C} \to \mathcal{C}$ and an object $I$ in $\mathcal{C}$, an $\Fun{F}$-algebra $(I^\infty, \phi_I)$ together with a morphism $\eta: I \to I^\infty$ is said to be a free $\Fun{F}$-algebra generated by $I$ if the following condition holds: for each $\Fun{F}$-algebra $(Q, q)$ and each morphism $f: I\to Q$ there is a unique morphism of $\Fun{F}$-algebras $f^\ast: (I^\infty, \phi_I)\to (Q,q)$ such that $f = f^\ast \circ \eta$:
\[
\begin{tikzcd}[column sep=1.5em]
\Fun{F}(I^\infty) \arrow{d}{\Fun{F}(f^*)} \arrow{r}{\phi_I} & I^\infty \arrow{d}{f^*} & I \arrow{l}{\eta}\arrow{dl}{f}\\
\Fun{F}(Q) \arrow{r}{q} & Q
\end{tikzcd}
\]
\end{definition}

An endofunctor $\Fun{F} : \mathcal{C} \rightarrow \mathcal{C}$ is a {\em varietor} \cite{Varieties} if there is always a free $\Fun{F}$-algebra generated by any $I \in \mathcal{C}$. 

Said in another manner, first note that for all endofunctors $\Fun{F}: \mathcal{C} \rightarrow \mathcal{C}$, there is a forgetful functor $U: Alg(\Fun{F}) \rightarrow \mathcal{C}$, which to every $\Fun{F}$-algebra $(C,a:\Fun{F}(C)\rightarrow C)$ associates $C$ in $\mathcal{C}$. 
Then $\Fun{F}$ is a varietor if the forgetful functor $U$ has a right adjoint. This right adjoint, when it exists, is the one that associates with each $C \in \mathcal{C}$ the free $\Fun{F}$–algebra on $C$.




The algebraic view of deterministic transition systems naturally leads to varietors. Consider again the functor $\Fun{F}$ which to any set $X$ associates $X\times Ac$, and a set $S$. Then the free algebra generated by $S$ exists and is given by formula 
\begin{equation}
S^\infty = S\ \cup\ \bigcup_{n>0} S \times Ac^n
\label{freealgdet}
\end{equation}
\noindent and morphism $\Phi_I: \ S^\infty\times Ac \to S^\infty$ sends a triple that consists of a state $s$, a sequence of actions $\la a_1 a_2 ... a_n \ra$, and another action $a$ to a pair $(s, \la a_1 ... a_n a \ra)$. 

An interesting consequence of Definition \ref{def:free} is that, for a varietor $\Fun{F}$, for any $\Fun{F}$-algebra $q: \ \Fun{F}(I) \rightarrow I$, there exists a unique morphism in $\mathcal{C}$ such that the following diagram commutes:

\[
\begin{tikzcd}[column sep=1.5em]
\Fun{F}(I^\infty) \arrow{d}{\Fun{F}(g)} \arrow{r}{\phi_I} & I^\infty \arrow{d}{g} & I \arrow{l}{\eta}\arrow{dl}{Id}\\
\Fun{F}(I) \arrow{r}{q} & I
\end{tikzcd}
\]

Consider the following maps: 
$$
g_i=q \circ \Fun{F}(q) \circ \ldots \circ \Fun{F}^{i-2}(q)\circ \Fun{F}^{i-1}(q): \Fun{F}^i(I) \rightarrow I
$$
In the case of the deterministic transition systems above, it is easy to see that $g_i$ acts as follows:
$$
g_i(s,\langle a_1 a_2\ldots a_i \rangle)=q(q(\ldots q(s,a_1),a_2),\ldots,a_i)
$$

This shows the usefulness of the notion of free algebra, which makes it possible to iterate on the transition function (in the case of deterministic transition systems), i.e. to take the ``most universal'' algebra that ``completes'' the original set of states with all states that are reachable by taking a finite (but unbounded) number of transitions. 

Indeed, this is a construction of a least fixed point, applied to the ``one-step'' transition function defined by the functor $\Fun{F}$. Classical set-based (or poset-based) semantics use Tarski theorem to prove the existence of such fixed points, and Kleene theorem when applicable, to give explicit formulas for least fixed points, under some mild hypotheses. Indeed, we recognize in Equation (\ref{freealgdet}) a Kleene type of formula.  

As we will not use endofunctors on $\Set$, but on more general categories, and in particular simplicial complexes (Section \ref{sec:simplicialcomplex}) and simplicial sets (Section \ref{sec:simplicialset}) we need some generalization of these theorems, that we give below (first a Tarski-like one, and then a Kleene-like one). 


\begin{lemma}[see \cite{trnkova1975free}]
\label{freemonos}
Given a functor $\Fun{F}: \mathcal{C}\to \mathcal{C}$ that preserves monomorphisms, the following is equivalent for any object $I$ in $\mathcal{C}$:
\begin{itemize}
    \item There exists a free $\Fun{F}$-algebra generated by $I$;
    \item There exists an object $X$ isomorphic to $I + \Fun{F}X$.
\end{itemize}
\end{lemma}

A statement more precise than Lemma \ref{freemonos} can be obtained for an endofunctor $\Fun{F}$ that commutes with colimits.

\begin{theorem}
\label{thm:thm1}
Let $\Fun{F}: \mathcal{C}\to \mathcal{C}$ be a functor that commutes with colimits and that preserves monomorphisms. Then for any object $I \in \mathcal{C}$ there exists a free $\Fun{F}$-algebra generated by $I$, which is given as follows. Let $I^\infty$ be the infinite coproduct:
$$
I^\infty=I \coprod \Fun{F}(I) \coprod \ldots \coprod \Fun{F}^n(I) \coprod \Fun{F}^{n+1}(I) \ldots
$$
As $\Fun{F}$ commutes with colimits, $\Fun{F}(I^\infty)$ is isomorphic to $\coprod\limits_{i\geq 1} \Fun{F}^i(I)$ and we denote by $\phi_I$ the map induced by the identities on each $\Fun{F}(\Fun{F} ^{n-1}(I))=\Fun{F}^n(I) \subseteq \Fun{F}(I^\infty)$ to $\Fun{F}^n(I) \subseteq I^{\infty}$. 
\end{theorem}

\begin{proof}
We use the free algebra construction of \cite{trnkova1975free} and Lemma \ref{freemonos}.
\end{proof}

We are now going to exploit this generalized semantical framework to recognize a pattern similar to deterministic transition systems in the definition and use of protocol complexes, modeling certain distributed machines, that we are introducing in next section. 

\section{Background on simplicial complexes for distributed computing and epistemic logic} 
\label{sec:backgroundsimplicialcomplex}
\label{sec:roundbased}

This section presents background on distributed computing and its relation to combinatorial topology (simplicial complexes and simplicial sets) as well as epistemic logics.
It focuses on models where computation can be organized in \emph{layers}~\cite{MosesR2002}, including asynchronous and synchronous models where processes (agents in distributed computing systems) can fail by crashing.

\subsection{Chromatic simplicial complexes and distributed computation} 
\label{sec:simplicialcomplex}

There is an intimate relationship between distributed computing and (combinatorial) topology, that has been thoroughly studied~\cite{herlihy}. The basic idea is to represent  the  final global states of a protocol as a simplicial complex, notion that we now recap below. 

\paragraph{Simplicial complexes}
A \emph{simplicial complex}~$K = (V, S)$ consists of a set $V$, together with a family $S$ of non-empty finite subsets of $V$ such that for all $X \in S$, $Y \subseteq X$ implies $Y \in S$.
An element of~$V$ is called a \emph{vertex} (plural vertices), and an element of $S$ is a \emph{simplex} (plural simplices).
When two simplices $X,Y \in S$ are such that $Y \subseteq X$, we say that $Y$ is a \emph{face} of $X$. Simplices that are maximal w.r.t.\ inclusion are called \emph{facets}.
The \emph{dimension} of a simplex $X \in S$ is  $|X|-1$, and a simplex of dimension~$n$ is called an \emph{$n$-simplex}.
Simplices of dimension 0, 1 and 2 are also called vertices, edges and triangles, respectively. Thus, we often identify a vertex $v \in V$ with the 0-simplex $\{v\} \in S$.

A simplicial complex $K$ is \emph{pure} if all its facets are of the same dimension, $n$.
In this case, we say that $K$ is of dimension~$n$. Given a set~$A$ of processes (that we will represent as colors), a \emph{chromatic simplicial complex} $K = (V, S, \chi)$ consists of a simplicial complex $(V,S)$ together with a coloring map $\chi : V \to A$, such that for all $X \in S$, all the vertices of $X$ have distinct colors.

Let $K$ and $K'$ be two simplicial complexes. A \emph{simplicial map} (a.k.a.\ \emph{morphism}) $f: K \rightarrow K'$ sends the vertices of $K$ to vertices of $K'$, such that if $X$ is a simplex of $K$, then $f(X)$ 
is a simplex of~$K'$.
When $K$ and~$K'$ are chromatic, a \emph{chromatic simplicial map} must additionally preserves colors.

\paragraph{Protocol complexes and task solvability.}

In the formalization of distributed systems of e.g.~\cite{herlihy}, a system comprised of~$n$ processes running concurrently is modeled by a pure chromatic simplicial complex of dimension~$n$, called the \emph{protocol complex}.
Each global state corresponds to a facet, and each vertex of this facet corresponds to the local states of one process.
Thus, each vertex is labeled with a process id, which is why this construction gives rise to a chromatic simplicial complex.
The remarkable connection with topology comes from the fact that, when processes communicate with each other, 
and the simplicial complex of global states evolves over time, some topological invariants are preserved.
These topological invariants in turn determine the computational power of the system.

The area of fault-tolerant distributed computing is concerned with determining which tasks are solvable in various computational models.
In the topological approach of~\cite{herlihy}, the following notion of task solvability is used.
A task $\mathcal{T}$ is given by the following data:
\begin{itemize}
\item a chromatic simplicial complex of inputs $I$, called the \emph{input complex};
\item a chromatic simplicial complex of outputs $O$, called the \emph{output complex};
\item a \emph{task specification} $T$, which is a sub-complex of $I\times O$. This chromatic simplicial complex~$T$ relates inputs with potential outputs, that any protocol solving the task should choose from. Naturally, $T$ comes with a projection morphism $\pi_T: \ T \rightarrow I$.
\end{itemize}

Given a particular architecture of communication, and an input complex $I$, we can produce the \emph{protocol complex} that consists of all global states that are reachable from an input of $I$ after one round of communication.
When we consider several rounds of communication, this is called the \emph{iterated} protocol complex: the complex of global states at the end of the first round of communication, can be used as the input complex from which the second round of communication is performed.

Given such a (possibly iterated) protocol complex $P$, because of its construction (which is formalized in Section \ref{sec:protfunctor}) there is always a simplicial map $\pi_P : P \to I$ that relates every global state of~$P$ (i.e., every simplex) to the inputs of~$I$ that it was computed from.

We can now define the notion of task solvability. 
A task ${\mathcal{T}}$ is \emph{solvable} using the (possibly iterated) protocol complex~$P$ if there exists a chromatic simplicial morphism $\delta: P\rightarrow T$ such that $\pi_{I}\, \circ\, \delta=\pi_P$, i.e., the diagram of simplicial complexes below commutes.

\begin{center}
\begin{tikzpicture}
  \node (s) {$P$};
  \node (xy) [below=2 of s] {${T}$};
  \node (x) [left=of xy] {$I$};
  \draw[<-] (x) to node [sloped, above] {$\pi_P$} (s);
  \draw[->, right] (s) to node {$\delta$} (xy);
  \draw[->] (xy) to node [below] {$\pi_I$} (x);
\end{tikzpicture}
\end{center}

The intuition behind this definition is the following. A facet $X$ in $P$ corresponds to a global state that is reachable from input $\pi_P(X)$ in $I$. The morphism~$\delta$ sends~$X$ to a facet $\delta(X)=(X_{I},X_{O})$ of~$T$: crucially, every vertex of $X$ (that is, every local state of a process) is assigned an output value in $X_O \subseteq O$.
The commutativity of the diagram expresses the fact that both $X$ and $\delta(X)$ correspond to the same input assignment $X_I = \pi_P(X)$.
The fact that $\delta(X) = (X_{I},X_{O})$ is a simplex of $T$ ensures that the chosen set of outputs is compatible with the set of inputs according to the task specification.
Finally, the fact that simplicial maps act on vertices means that a process must decide its output value based only on its local state.
When one vertex (local state) belongs to two or more simplices (global states), this means that the two global states are \emph{indistinguishable} from the point of view of this process: it must choose the same output value in both situations.



\subsection{Chromatic simplicial complexes as simplicial models for epistemic logics} 

The usual $\mathbf{S5_n}$ epistemic logic for a multi-agent system has the following syntax. 
Let $\I{AP}$ be a countable set of atomic propositions and $A$ a finite set of agents. The language of epistemic logic formulas $\mathcal{L}_\cK(A,\I{AP})$, or just~$\cL_\cK$ if~$A$ and~$\I{AP}$ are left implicit, is generated by the following BNF grammar:
\[
\varphi ::= p \mid \neg\varphi \mid (\varphi \land \varphi) \mid
K_a\varphi \qquad p \in \I{AP},\ a \in A
\]

We are also interested in epistemic logic with \emph{common knowledge} operators. We write $\cL_{\cC\cK}(A,\I{AP})$ for the language of these formulas, defined by:
\[
\varphi ::= p \mid \neg\varphi \mid (\varphi \land \varphi) \mid
K_a\varphi \mid C_B\varphi \qquad p \in \I{AP},\ a \in A,\ B \subseteq A
\]

Recall that usually, an $\mathbf{S5_n}$ epistemic model is based on a \emph{Kripke frame}, a graph whose nodes are the possible worlds, and whose edges are labeled with agents that cannot distinguish between two worlds.
A \emph{Kripke model} is obtained by labeling the nodes of the graph with sets of atomic propositions.
Much work exists on using Kripke models to reason about distributed computing~\cite{FHMVbook}, we are more interested here instead on using simplicial complex models instead.

An epistemic logic framework to model the topology approach to distributed computing  was presented in~\cite{infcomp}.
To do so, Kripke models are replaced by a dual notion of \emph{simplicial complex models.}
In fact, there is an equivalence of categories between Kripke models and simplicial models.
Thus, simplicial models retain the  properties of Kripke models, such as soundness and completeness w.r.t.\  to the logic $\mathbf{S5_n}$.
Furthermore, this simplicial approach to epistemic logic can be extended to \emph{dynamic epistemic logic} (DEL)~\cite{sep-dynamic-epistemic,DEL:2007}, using a simplicial version of \emph{action models}~\cite{baltagMS:98}.

\subsection{Chromatic (augmented semi-)simplicial sets} 
\label{sec:simplicialset}

Unfortunately, chromatic simplicial complexes are rather inconvenient in categorical terms, whereas chromatic (semi-)simplicial sets are both more general and much better behaved. Chromatic semi-simplicial sets, or csets, have already been studied in the context of distributed computing in~\cite{LICS2023}.

Let us start with an intuitive explanation of what is a chromatic semi-simplicial set.
Recall that, to define a simplicial complex, one must provide a set~$V$ of vertices; and then an $n$-simplex of higher dimension is identified with its set of $n+1$ vertices.
In contrast, to define a semi-simplicial set, we directly give the set $S_n$ of $n$-simplices for every~$n$.
To describe how those basic blocks are ``glued together'' into a geometric structure, we moreover need to define face maps of the form ${\partial^i_n : S_n \to S_{n-1}}$, mapping every $n$-simplex to its $(n-1)$-dimensional faces.
In fact, we do something a little bit more unusual here: we add colors to the definition.
Instead of having a set $S_n$ of $n$-simplices for every dimension~$n$, we will have a set $X(U)$ of $U$-colored simplices for every set of colors $U \subseteq A$.
This yields the definition of \emph{chromatic semi-simplicial sets}.

We start by giving a slick formal definition of the category of csets as a category of presheaves. We will then unpack the data to give a more concrete definition.
As with the site of semi-simplicial sets (see e.g., \cite{ss,Riehl1}), the site of chromatic augmented semi-simplicial sets is a posetal category $\Gamma$, which is the category of subsets of $A$ with the inclusion partial order, defined below.

\begin{definition}[The category \(\Gamma\)]
    \label{def:gamma}
    Given a set $A$ of processes, the category $\Gamma$ is such that:
    \begin{itemize}
        \item Its objects are (possibly empty) subsets of~$A$.
        \item There is a unique morphism $\delta_{U,V} : U \to V$ in~$\Gamma$ whenever ${U \subseteq V}$. Composition $\delta_{V,W} \circ \delta_{U,V} = \delta_{U,W}$ is given by the fact that ${U \subseteq V \subseteq W}$ implies $U \subseteq W$.
    \end{itemize}
\end{definition}
The category of \emph{chromatic augmented semi-simplicial sets}, or cset for short, is the category $\Simp$ of presheaves on $\Gamma$.
Hence a cset is defined as a functor $X : \Gamma^{\op} \to \Set$, and morphisms of csets are natural transformations between such functors.

Given a cset $X \in \Simp$, and a group of processes $U \subseteq A$, the elements of $X(U)$ are called the \emph{$U$-simplices}.
When there is no ambiguity, we write $\partial_{U,V} : X(V) \to X(U)$ for the boundary operator $X(\delta_{U,V})$.
For $x \in X(V)$ a $V$-simplex, $\partial_{U,V}(x)$ is called the $U$-face of $x$. 
If it is clear which $V$ is considered, we simply write $\partial_U(x)$.
The \emph{standard $U$-simplex} $\Gamma[U]$, for some set of processes $U \subseteq A$, is defined as the representable presheaf $\Gamma(-,U)$, image of~$U$ by the Yoneda embedding $y: \Gamma \to \Simp$.

Let us unpack these definitions.
A cset $X$ consists of a family of sets $X(U)$ for every $U \subseteq A$, together with a family of maps $\partial_{U,V} : X(V) \to X(U)$ for every $U \subseteq V \subseteq A$.
Intuitively, from any $V$-simplex (i.e., a simplex whose vertices are colored by $V$), we can extract its $U$-face by keeping only the vertices whose color belongs to $U$.
These face maps must moreover be compatible, in the sense that if $U \subseteq V \subseteq W$, then $\partial_{U,V} \circ \partial_{V,W} = \partial_{U,W}$.
The standard $U$-simplex $\Gamma[U]$ has a unique $T$-simplex for every $T \subseteq U$; and no $T$-simplex for $U \subsetneq T$. That is, $\Gamma[U](T)$ is a singleton set for each $T \subseteq U$, and is the empty set otherwise.
The face maps of $\Gamma[U]$ are trivial since all sets are singletons.
Geometrically, $\Gamma[U]$ can be viewed as a single simplex with vertices colored by~$U$, and all of its faces.

\begin{example}
\label{ex:cset-base}
A cset~$X$ is depicted in the figure below. We represent it as a semi-simplicial set together with colors on the vertices.
Let us write $A = \{r,g,b\}$ for the set of colors (red, green, blue).
The two elements of $X(\{r,g,b\})$ are depicted as triangles, whose vertices colored in red, green and blue.
The sets $X(\{r,g\})$ and $X(\{g,b\})$ both contain three elements, depicted as edges (with the corresponding colors on the vertices); while the set $X(\{r,b\})$ only contains two elements (the two parallel red/blue edges).
The sets $X(\{r\})$, $X(\{g\})$ and $X(\{b\})$ contain respectively two, three, and two elements; depicted as red, green and blue vertices.
Finally, the set $X(\emptyset)$ contains two elements, represented as dashed regions (they can be interpreted as a notion of generalized connected components).
The face maps $\partial_U$ for $U \subseteq \{r,g,b\}$ give the equations that permit to glue together all of these simplices, along lower dimensional simplices.
%
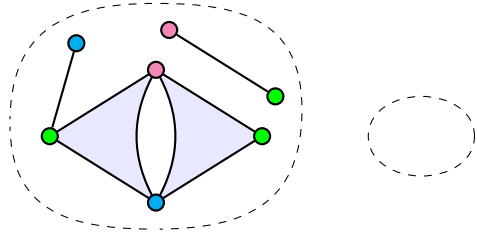
\begin{center}
    \begin{tikzpicture}[scale=1.4]
	\begin{pgfonlayer}{nodelayer}
		\node [style=green vertex] (3) at (2, 0) {};
		\node [style=cyan vertex] (4) at (-1.5, 1.75) {};
		\node [style=magenta vertex] (5) at (0.25, 2) {};
		\node [style=green vertex] (6) at (2.25, 0.75) {};
		\node [style=magenta vertex] (0) at (0, 1.25) {};
		\node [style=cyan vertex] (1) at (0, -1.25) {};
		\node [style=green vertex] (2) at (-2, 0) {};
		\node [style=empty node] (8) at (0, 2.5) {};
		\node [style=empty node] (9) at (-2.75, 0.25) {};
		\node [style=empty node] (10) at (0, -1.75) {};
		\node [style=empty node] (11) at (2.75, 0.5) {};
		\node [style=none] (12) at (4, 0) {};
		\node [style=none] (13) at (5, 0.75) {};
		\node [style=none] (14) at (6, 0) {};
		\node [style=none] (15) at (5, -0.75) {};
	\end{pgfonlayer}
	\begin{pgfonlayer}{edgelayer}
		\draw [style=thick edge] (5) to (6);
		\draw [style=thick edge] (4) to (2);
		\draw [style=filled path] (0.center)
			 to (3.center)
			 to (1.center)
			 to [bend right] cycle;
		\draw [style=filled path] (0.center)
			 to (2.center)
			 to (1.center)
			 to [bend left] cycle;
		\draw [style=dashed] (9.center)
			 to [in=-180, out=-90, looseness=1.25] (10.center)
			 to [in=-90, out=0, looseness=1.25] (11.center)
			 to [in=0, out=90, looseness=1.25] (8.center)
			 to [in=90, out=-180, looseness=1.25] cycle;
		\draw [style=dashed] (14.center)
			 to [in=0, out=90] (13.center)
			 to [in=90, out=-180] (12.center)
			 to [in=-180, out=-90] (15.center)
			 to [in=-90, out=0] cycle;
	\end{pgfonlayer}
\end{tikzpicture}

\end{center}
\end{example}


In~\cite{LICS2023}, csets have been used to model distributed computing protocols.
As with simplicial complexes, the vertices represent local states of individual agents (processes) $a \in {A}$, and its color, representing the name of the process.
Maximal simplexes (i.e., facets) correspond to global states of the system.
What is new with the semi-simplicial set approach is that we can also explicitly model information in-between the two: $U$-simplexes represent the combined states of a set~$U$ of participating processes.
This is closely linked to the notion of \emph{distributed knowledge}, where a set of agents is viewed as a single entity that can collectively know some facts about the possible worlds.

\section{Protocol complex functors for round-based oblivious distributed systems} 
\label{sec:protocolfunctor}

The protocol complex that appears in fault-tolerant distributed computing can be formally defined a particular functorial construction, at least for round-based oblivious systems, as we now describe.
The advantage of our functorial definition is that it highlights the modular structure of the protocol complex: one only needs to specify what happens to a single simplex (i.e., global state), for one round of computation.
From this, the protocol complex can be derived for any input complex and any number of rounds.
This can be constructed, and studied, under the following (informal for now) assumptions.
We hint at further generalizations in Section \ref{sec:beyond}.

\begin{enumerate}
\item Communication is \emph{round-based}. This means that protocols are structured as a sequence of rounds, where processes communicate with each other, potentially with faults occurring.
Note that this does not necessarily imply synchrony: for instance, the immediate snapshot model of Section~\ref{ex:async} exhibits a round structure, even though it can be implemented in a fully asynchronous way.
Indeed, some processes are allowed to start subsequent rounds even though some other processes are still in the previous rounds; and they do so in a fresh memory array to avoid interfering.
In such cases, a slow process is indistinguishable from a crashed process.

\item Communication is \emph{oblivious}. We assume that the communication patterns that can occur do not depend on the values of local states of processes. This means that the ``set'' (actually a simplicial complex or a cset) of states that are reachable from a given global state only depends on the set of participating processes, and not on the actual values held by the processes. 

\item Communication is \emph{wait-free}.
At any point in time, we always consider the possibility that only a subset of the currently active processes will participate in the next round (either because the other processes are slow, or have crashed).
So if a set~$A$ of processes are active, and ${B \subseteq A}$, 
performing one round of communication within the set~$A$ must also include the effect of communication on the smaller set of processes~$B$.

\end{enumerate}

With these hypotheses, protocol complexes play the role of automata in our formalization, as shown by the algebraic formulation that we develop now. For distributed computing, we are dealing with algebras over an endofunctor on chromatic simplicial complexes (or chromatic simplicial sets, as explained in \Cref{sec:simplicialset}). This differs from the usual setting of sequential computation and ordinary automata, which are algebras on sets.
In logical and categorical terms, this is similar though, since csets form a topos which are a categorical generalization of sets \cite{topos}.

\subsection{The protocol complex functor} 

Hypotheses 1 and 2 above indicate that we only need to specify the protocol complex for a single initial state and a single round of communication.
Thus, all relevant information should lie within a functor $F: \Gamma \rightarrow \Gamma^{\op} \Set$, where the objects of $\Gamma$ encode the possible sets of participating processes.
Hypothesis 3 indicates that this functor $F$ should be monotonic, in the sense that $F$ preserves monomorphisms (both in $\Gamma$ and in $\Gamma^{\op} \Set$, monomorphisms are inclusions).

Once such a $F$ is defined, we can extend it in a canonical way to a functor ${F_{!} : \Simp \to \Simp}$ that acts on an input complex, consisting of many global states glued together.
This extension of the functor $F$ is a classic categorical construction called a Yoneda extension, that we recall below.


\begin{lemma}[see \cite{PGMZCalculus}]
\label{fundamental}
\label{basic}
Let $B$ be a small category, $C$ a category. Every functor
$F: B \rightarrow C $
defines a functor
${F^*: C \rightarrow B^{\op} \Set}$
defined on the objects $c$ in $C$ by $F^*(c)=\Hom[C]{F-}{c}$.
If $C$ has small direct limits (also called filtered limits in \cite{SMLCategories})
then $F^*$ has a left adjoint
$ F_{!}: B^{\op} \Set \rightarrow C $
which is the unique\footnote{$F_{!}$ is unique since all elements of $B^{\op} \Set$ are direct limits of representable functors, i.e. of the form $B[b]$.} functor
commuting with direct limits such that
$F_{!}(B[b])=F(b)$ for all $b \in B$ (where, as before, the notation $B[b]$ denotes the representable functor $\Hom[B]{-}{b}$).

Moreover, there is an equivalence of categories between the category of functors from $B$ to $C$, and the full subcategory of
the category of functors from $B^{\op} \Set$ to $C$ consisting of those functors
that commute with direct limits.
\end{lemma}

When particularizing Lemma~\ref{basic} to the context of chromatic augmented simplicial sets and distributed computing, it reads as follows. Consider a functor $\Fun{F}: \Gamma \rightarrow \Simp$ specifying the way the one-round protocol complex operator acts on a single global state with a set $U \subseteq A$ of participating processes. Then $\Fun{F}_{!}: \Simp \rightarrow \Simp$ extends $\Fun{F}$ to all chromatic augmented semi-simplicial sets (and not just the standard $U$-simplex $\Gamma[U]$ alone), commutes with colimits, and therefore defines the effect of one-round communication for any set of combined global states.
Note that even though we are mainly interested in the functor $F_{!}$, Lemma~\ref{basic} also gives ``for free'' a right adjoint to $F_{!}$, called $F^*$. 
We will not make use of the functor $F^*$ to study protocol complexes, but it might be worth studying what this functor means in the context of distributed computing.




\subsection{A first example: the asynchronous immediate snapshot model} 

\label{ex:async}

Let us consider processes that communicate by reading and writing asynchronously into a shared memory. 
Let us assume that the set of participating processes is $A = \{ a_1, \ldots, a_n \}$.
Each process $a_i$ holds a private value called its a local state, $v_i$.
The shared memory consists of a sequence of shared arrays of size~$n$, one for each round. Thus, at round~$k$, we use a fresh array $x^k$ consisting of~$n$ memory locations $x^k_1$, \ldots, $x^k_n$.
Process $a_i$ can write in its own memory cell $x^k_i$, but it is allowed to read all memory cells.
In this model, each process $a_i$ does the following, at round~$k$:
\begin{itemize}
\item Process $a_i$ writes its current local state in the shared array:\; $x^k_i \leftarrow v_i$.
\item It then takes a snapshot of the entire memory, reading the current values of $x^k_1$, \ldots, $x^k_n$.
\item Then, it performs some local computation to update its local state according to the values observed in the round:\; $v_i \leftarrow f(v_i, x^k_1, \ldots, x^k_n)$, for some computable function~$f$.
\end{itemize}

All processes perform these steps asynchronously, for each round from $1$, \ldots, $k$.
A protocol is called \emph{full-information} when the local computation in the third step simply consists in storing the entire array $x^k_1$, \ldots, $x^k_n$ that has been read.
An execution is called \emph{immediate} when the first two steps (write and snapshot) never occur too far apart from each other. More precisely, whenever a process takes a snapshots, all processes with pending writes must also take a snapshot before another write can occur. This immediacy condition somewhat breaks the asynchrony assumption; but this behavior can actually be simulated in a fully asynchronous setting~\cite{BorowskyG92}.

In our setting, the full-information immediate snapshot protocol can be formally defined by the following functor $G : \Gamma \to \Simp$.
For $A \in \Gamma$ a set of active processes, we define a chromatic simplicial set (in fact, it happens to be a simplicial complex) $G(A) \in \Simp$ as follows.
The vertices of $G(A)$ are pairs of the form $(b_i,V_i)$, where $b_i$ is a participating process, and $V_i$ is the \emph{view} of $b_i$, that is, the set of processes whose value was read by $b_i$.
The $B$-simplices of $G(A)$ are compatible sets of such vertices.
Formally, if $B \subsetneq A$, we let $G(A)(B) = \emptyset$.
Otherwise, let us write  $B=\{b_0,\ldots,b_m\}\subseteq A$, and define the set of $B$-colored simplexes of $G(A)$:
\begin{equation}
    G(A)(B)=\left\{\{(b_0,V_0),\cdots,(b_m,V_m)\} 
\left|\begin{array}{l}
V_i \subseteq A \\
\forall i, \ 0 \leq i \leq m, \ b_i \in V_i \\
\forall i,j, \ 0 \leq i, j \leq m, \ (V_i \subseteq V_j) \mbox{ or } (V_j \subseteq V_i) \\
\forall i, j, \ 0 \leq i, j \leq m, \ a_i \in V_j \implies V_i \subseteq V_j \\
\end{array}\right.\right\}
\label{eq:async}
\end{equation}
(see e.g.~\cite{updatescan1} or~\cite{updatescan2} for why those particular conditions on the views $V_i$ work).
%
%
%
To complete the definition of the simplicial set $G(A)$, we also need to define the face maps: for $B' \subseteq B$, 
\[\begin{array}{rlcl}
G(A)(\delta_{B',B}): & G(A)(B) & \longrightarrow & G(A)(B') \\
& \{(b_0,V_0),\cdots,(b_m,V_m)\} & \longmapsto & 
\{(b_{i},V_{i}) \mid b_i \in B'\}
\end{array}\]

This defines the functor $G$ on the objects of $\Gamma$.
Finally, the action of $G$ on the inclusion maps $\delta_{A',A}$ is given by:
\[\begin{array}{rlcl}
G(\delta_{A',A})(B): & G(A')(B) & \longrightarrow & G(A)(B) \\
& \{(b_0,V_0),\cdots,(b_m,V_m)\} & \longmapsto & 
\{(b_{0},V_{0}),\ldots,(b_{m},V_{m})\}
\end{array}\]
which is well-defined because if $V_i \subseteq A'$ and $A' \subseteq A$, then $V_i \subseteq A$. 



This completes the definition of the functor $G : \Gamma \to \Simp$, which describes the protocol complex of the immediate snapshot protocol for one round and one initial state.
From Lemma \ref{fundamental}, we get a functor $G_! : \Simp \to \Simp$, which describes the one-round protocol complex for any input complex.
Intuitively, what the functor $G_!$ does is the following: for every simplex of the input complex, whose set of participating processes is $U \subseteq A$, we take one copy of $G(U)$; and then those copies are glued together according to the face maps of the input complex.


\begin{example}
Below is depicted the complex $G(A)$, where the set of processes is $A = \{a, b, c\}$ (also represented as colors blue, red, green, respectively).
The complex $G(A)$ consists of $12$ vertices labeled $v_0, \ldots, v_{11}$, and the encoding of Equation (\ref{eq:async}) for each vertex is given next to the picture.
The $13$ triangles correspond to the $13$ sets of compatible vertices respecting the conditions of Equation (\ref{eq:async}); they also correspond to the $13$ executions of the immediate snapshot protocol for three processes.

Let $A' = \{a,b\} \subseteq A$. The simplicial complex $G(A')$ can also be seen in the picture, as the sub-complex of $G(A)$ consisting of the four vertices $v_0, v_1, v_6, v_8$, together with the three edges between them.
The morphism of csets $G(\delta_{A',A}) : G(A') \to G(A)$ is the inclusion of $G(A')$ into $G(A)$.

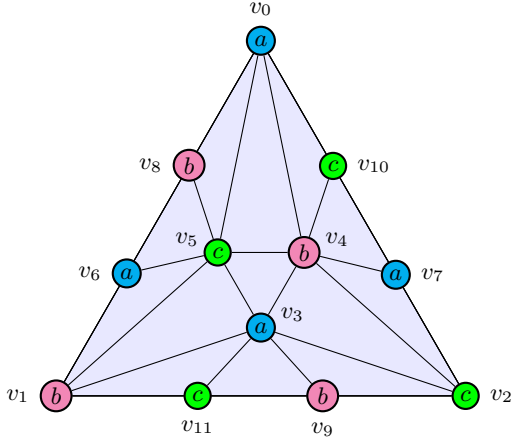
\begin{minipage}{0.48\textwidth}
  \begin{center}
  \begin{tikzpicture}[scale=0.9, font=\footnotesize,
  cloud/.style={draw=black,thick,circle,fill=cyan,inner sep=1.5pt,minimum size=6pt}, cloudblack/.style={draw=black,thick,circle,fill=magenta!60,inner sep=1.5pt,minimum size=6pt},  cloudgrey/.style={draw=black,thick,circle,fill=green,inner sep=1.5pt,minimum size=6pt}]
  \filldraw[fill=blue!60, fill opacity=0.15] (0,0) -- (6,0) -- (3,5.2) -- cycle;
  \node[cloudblack, label=left:{$v_1$}] (p) at (0,0) {$b$};
  \node[cloudgrey, label=right:{$v_2$}] (q) at (6,0) {$c$};
  \node[cloud, label=above:{$v_0$}] (r) at (3,5.2) {$a$};

  \draw (p) -- (q) node[cloudblack, label=below:{$v_9$}](p-pq) [pos=0.66] {$b$}
                   node[cloudgrey, label=below:{$v_{11}$}](q-pq) [pos=0.33] {$c$}; 
  \draw (p) -- (r) node[cloudblack, label=left:{$v_8$}](p-pr) [pos=0.66] {$b$}
                   node[cloud, label=left:{$v_6$}](r-pr) [pos=0.33] {$a$}; 
  \draw (q) -- (r) node[cloudgrey, label=right:{$v_{10}$}](q-qr) [pos=0.66] {$c$}
                   node[cloud, label=right:{$v_7$}](r-qr) [pos=0.33] {$a$}; 
  \node[cloudblack, label={[xshift=-2pt,yshift=5pt]right:{$v_{4}$}}] (p-pqr) at (4.5-5.2/6*1,2.6-1/2*1) {$b$};
  \node[cloudgrey, label={[xshift=2pt,yshift=5pt]left:{$v_{5}$}}] (q-pqr) at (1.5+5.2/6*1,2.6-1/2*1) {$c$};
  \node[cloud, label={[xshift=-2pt,yshift=5pt]right:{$v_3$}}] (r-pqr) at (3,1) {$a$};
  \draw (p) -- (q-pqr) -- (r-pqr) -- (p-pqr) -- (q) -- (r-pqr) -- (p);
  \draw (r) -- (q-pqr) -- (p-pqr) -- (r);
  \draw (p-pq) -- (r-pqr) -- (q-pq);
  \draw (p-pr) -- (q-pqr) -- (r-pr);
  \draw (q-qr) -- (p-pqr) -- (r-qr);
  \draw (p) edge node[auto] {} (r-pr); 
  \draw (r-pr) edge node[auto] {} (p-pr); 
  \draw (p-pr) edge node[auto] {} (r); 
  \draw (r) edge node[auto] {} (q-qr); 
  \draw (q-qr) edge node[auto] {} (r-qr); 
  \draw (r-qr) edge node[auto] {} (q);
  \draw (q) edge node[auto] {} (p-pq);
  \draw (p-pq) edge node[auto] {} (q-pq); 
  \draw (q-pq) edge node[auto] {} (p);
  \end{tikzpicture}
  \end{center}
\end{minipage}
\begin{minipage}{0.48\textwidth}
  \begin{center}
  \begin{align*}
  &v_0 = (a, \{a\}) & v_6 &= (a, \{a,b\})\\
  &v_1 = (b, \{b\}) & v_7 &= (a, \{a,c\})\\
  &v_2 = (c, \{c\}) & v_8 &= (b, \{a,b\})\\
  &v_3 = (a, \{a,b,c\}) & v_9 &= (b, \{b,c\})\\
  &v_4 = (b, \{a,b,c\}) & v_{10} &= (c, \{a,c\})\\
  &v_5 = (c, \{a,b,c\}) & v_{11} &= (c, \{b,c\})
  \end{align*}
  \end{center}
\end{minipage}
\end{example}

\subsection{A more general class of examples: dynamic networks} 

\label{sec:dynamicnetwork}

The immediate snapshot example is an instance of a more general framework, called \emph{dynamic networks}~\cite{Kuhn11dynamic}, also studied in the setting of epistemic logic using the so-called \emph{communication pattern models}~\cite{patternmodels}.
Here, we describe the special case of dynamic networks with \emph{oblivious adversary}~\cite{CouloumaGP15}, meaning that the set of possible communication graphs is the same at each round. A hint of a potential generalization is given in Section \ref{sec:beyond} and Example \ref{ex:dynamicnetworks}. 
Moreover, we model crash failures explicitly by using the lack of reflexive edge in a communication graph to mean that the corresponding agent had crashed during the current round of communication.

\begin{definition}[Communication graph]
\label{def:communication-graph}
Given a finite set of agents~$A$, a \emph{communication graph} on $A$ is a (directed) graph $G \subseteq A \times A$.
When $G$ is clear from context, we write $a \to b$ instead of $(a,b) \in G$.
We say that agent~$a$ is \emph{active} in~$G$ when $a \to a$, and that~$a$ has \emph{crashed} otherwise.
The \emph{view} of an active agent~$a$ in~$G$ is $\view{G}{a} = \{ b \in A \mid b \to a \}$.
The view is undefined for crashed agents.
\end{definition}

Intuitively, a communication graph describes a round of communication between the currently active agents.
An edge $a \to b$ indicates that agent~$a$ successfully sent a message to agent~$b$.
Reflexive edges indicate whether an agent has crashed during the current round, or is still active: crashed agents will not participate in subsequent rounds.

\begin{definition}[Dynamic network model]
\label{def:dynamicnetwork}
Given a finite set of agents~$A$, an \emph{oblivious message adversary} for~$A$ is a set of communication graphs on~$A$.
A \emph{dynamic network model} is a family $(M_U)_{U \subseteq A}$, where for each subset $U \subseteq A$, $M_U$ is an oblivious message adversary for~$U$.
\end{definition}

An oblivious message adversary specifies a set of possible communicative events that may occur at each round.
For instance, if we want to specify that at most one message may be lost at each round, the corresponding adversary is the set of all communication graphs where at most one edge is missing.

On the other hand, a dynamic network model specifies an oblivious message adversary for each possible subset of the participating agents.
This allows us to model, for example, a situation where at most one agent may crash at each round of computation.
If a agent crashes during the first round, then in the second round we use the message adversary that concerns only the agent that are still active.

It may seem cumbersome to specify a message adversary $M_U$ for every possible subset of agents $U \subseteq A$.
Indeed, most protocols considered in practice behave ``uniformly'', meaning that the set of possible behaviors does not depend on the set of participating agents.
In such a case, one can specify the adversary $M_A$ for the set of all agents; and then the family $(M_U)$ can be defined canonically by $M_U = \{ G\!\restriction_{U} \mid G \in M_A \}$, where $G\!\restriction_{U}$ denotes the induced subgraph of $G$ restricted to the set of vertices~$U$.
Not all protocols behave like this however: a counter-example is given by so-called \emph{$k$-resilient protocols}, where at most~$k$ agents may crash during the execution of the system.
To model this, the message adversary $M_U$ allows crashes to occur only when $\#(A \setminus U) \le k$.

\paragraph{Full-information protocol complex for dynamic networks.} 

Let $A$ be a finite set of agents, and let us fix a dynamic network model $(M_U)_{U \subseteq A}$.
We describe the corresponding full-information protocol complex as a functor $F : \Gamma \to \Simp$.

With an object $U \subseteq A$ of $\Gamma$, we associate a cset, defined as a functor $F(U) : \Gamma^\op \to \Set$.
For $V \subseteq A$, we define $F(U)(V)$ the set of $V$-colored simplexes of the cset  $F(U)$.
\begin{align*}
F(U)(V) &= \emptyset & & \text{ if $V \not \subseteq U$}\\
F(U)(V) &= \{ \{(v, \view{G}{v}) \mid v \in V \} \mid G \in M_{U'} \text{ for some } U' \subseteq U, \text{ and } V \subseteq \activeProcs{G} \} & & \text{otherwise}
\end{align*}
To get a functor $F(U) : \Gamma^\op \to \Set$, we also need to define the face maps: if $V \subseteq W$, we define
\begin{align*}
F(U)(\delta_{V,W}) \quad:\quad &F(U)(W) &&\longrightarrow &&F(U)(V)\\
&\{ (w, \view{G}{W}) \mid w \in W \} &&\longmapsto &&\{ (v, \view{G}{V}) \mid v \in V \}
\end{align*}
which is well-defined since whenever $W \subseteq \activeProcs{G}$, we also have $V \subseteq W \subseteq \activeProcs{G}$.

Finally, $F$ itself needs to be a functor from $\Gamma$ to $\Simp$, so for $U \subseteq T$ we need to define a morphism of simplicial sets $F(\delta_{U,T}) : F(U) \to F(T)$.
This morphism is given by the family of functions:
\begin{align*}
F(\delta_{U,T})(V) \quad:\quad &F(U)(V) &&\longrightarrow &&F(T)(V)\\
&\{ (v, \view{G}{v}) \mid v \in V \} &&\longmapsto && \{ (v, \view{G}{v}) \mid v \in V \}
\end{align*}
which is well-defined, because the communication graph $G$ is required to be in $M_{U'}$ for some $U' \subseteq U$; but then we can still pick the same $G$ since $U' \subseteq U \subseteq T$.

%
%
%



\begin{example}[Reliable broadcast]
Let us start with a very simple and slightly degenerate example, that illustrates some of the subtleties in our definition of the protocol complex.
We consider a reliable broadcast, with no crashes or missed messages. 
This can be modeled as a dynamic network model, by taking for every $U \subseteq A$ the oblivious message adversary $M_U = \{ K_U \}$, where $K_U$ denotes the complete graph on the set of vertices~$U$.

One might expect the protocol complex functor to send a set $U \subseteq A$ of agents to a single simplex, since there is only one possible execution; but this is not the case.
Indeed, in the definition of the simplicial complex $F(U)$, we have to choose a subset $U' \subseteq U$ of participating agents.
Then, we (reliably) broadcast messages among this set~$U'$ of agents, and collect their views.
So, in fact, we get one maximal simplex for each subset $U' \subseteq U$.
This detail is crucial to be able to define the functorial action $F(\delta_{U,T})$, which tells us how the functor $F_!$ will ``glue together'' the various components of the protocol complex.

The figure below depicts simplicial complexes $F(\{a,b\})$ and $F(\{a,b,c\})$ for the reliable broadcast protocol.
On the picture, only the names of the agents ($a$, $b$, or $c$) are written; not their views.
Since $\{a,b\} \subseteq \{a,b,c\}$, we have a morphism $F(\delta_{\{a,b\}, \{a,b,c\}}) : F(\{a,b\}) \to F(\{a,b,c\})$, which is the obvious inclusion.

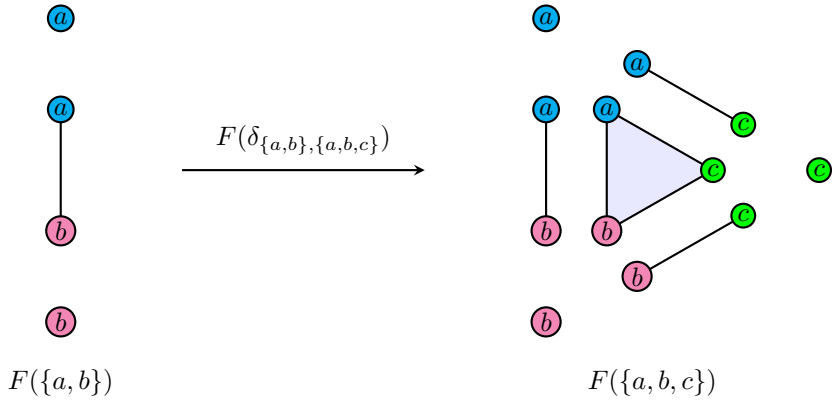
\begin{center}
	\begin{tikzpicture}[scale=1.6, auto, font={\small}]
	\begin{pgfonlayer}{nodelayer}
		\node [style=cyan vertex] (0) at (-1, 2) {$a$};
		\node [style=green vertex] (1) at (0.75, 1) {$c$};
		\node [style=magenta vertex] (2) at (-1, 0) {$b$};
		\node [style=none] (12) at (-1, 0) {};
		\node [style=none] (13) at (-1, 2) {};
		\node [style=none] (14) at (0.75, 1) {};
		\node [style=cyan vertex] (15) at (-2, 2) {$a$};
		\node [style=magenta vertex] (16) at (-2, 0) {$b$};
		\node [style=none] (17) at (-2, 0) {};
		\node [style=none] (18) at (-2, 2) {};
		\node [style=cyan vertex] (19) at (-0.5, 2.75) {$a$};
		\node [style=green vertex] (20) at (1.25, 1.75) {$c$};
		\node [style=none] (21) at (-0.5, 2.75) {};
		\node [style=none] (22) at (1.25, 1.75) {};
		\node [style=green vertex] (23) at (1.25, 0.25) {$c$};
		\node [style=magenta vertex] (24) at (-0.5, -0.75) {$b$};
		\node [style=none] (25) at (-0.5, -0.75) {};
		\node [style=none] (26) at (1.25, 0.25) {};
		\node [style=cyan vertex] (27) at (-2, 3.5) {$a$};
		\node [style=none] (28) at (-2, 3.5) {};
		\node [style=magenta vertex] (29) at (-2, -1.5) {$b$};
		\node [style=none] (30) at (-2, -1.5) {};
		\node [style=green vertex] (31) at (2.5, 1) {$c$};
		\node [style=none] (32) at (2.5, 1) {};
		\node [style=cyan vertex] (33) at (-10, 2) {$a$};
		\node [style=magenta vertex] (34) at (-10, 0) {$b$};
		\node [style=none] (35) at (-10, 0) {};
		\node [style=none] (36) at (-10, 2) {};
		\node [style=cyan vertex] (37) at (-10, 3.5) {$a$};
		\node [style=none] (38) at (-10, 3.5) {};
		\node [style=magenta vertex] (39) at (-10, -1.5) {$b$};
		\node [style=none] (40) at (-10, -1.5) {};
		\node [style=none] (41) at (-10, -2.5) {};
		\node [style=none] (42) at (-10, -2.5) {$F(\{a, b\})$};
		\node [style=none] (43) at (-0.25, -2.5) {};
		\node [style=none] (44) at (-0.25, -2.5) {$F(\{a, b, c\})$};
		\node [style=none] (45) at (-8, 1) {};
		\node [style=none] (46) at (-4, 1) {};
		\node [style=none] (47) at (-4, 1) {};
		\node [style=none] (48) at (-5.5, 1.5) {};
		\node [style=none] (49) at (-5.5, 1.5) {};
		\node [style=none] (50) at (-6, 1.5) {$F(\delta_{\{a,b\}, \{a,b,c\}})$};
	\end{pgfonlayer}
	\begin{pgfonlayer}{edgelayer}
		\draw [style=filled path] (13.center)
			 to (14.center)
			 to (12.center)
			 to cycle;
		\draw [style=filled path] (17.center) to (18.center);
		\draw [style=filled path] (21.center) to (22.center);
		\draw [style=filled path] (26.center) to (25.center);
		\draw [style=filled path] (35.center) to (36.center);
		\draw [style=arrow] (45.center) to (47.center);
	\end{pgfonlayer}
\end{tikzpicture}

\end{center}
\end{example}

\begin{example}[Immediate snapshot, revisited]

Let us consider again the immediate snapshot protocol of Example~\ref{ex:async}, seen as a dynamic network model.
First, we need to define the possible communication graphs for any set~$U$ of participating agents.
Define an \emph{ordered partition} of~$U$ to be a sequence of subsets $S = (S_1, \ldots, S_k)$, such that $\bigcup_{i} S_i = U$ and $S_i \cap S_j = \emptyset$ for all $i \neq j$.
The communication graph $G_S$ associated with an ordered partition has an edge $a \to b$ (for $a,b \in U$) whenever $a \in S_i$ and $b \in S_j$ with $i \leq j$.
\[
M_U = \{ G_{S} \mid S \text{ is an ordered partition of } U \}
\]
This defines the immediate snapshot protocol as a dynamic network model.
For example, for a set of three agents $U = \{a, b, c\}$, the above definition yields a total of $13$ graphs, corresponding to the $13$ ordered partitions of $\{a, b, c\}$.
Four graphs are depicted below, with the corresponding ordered partition written under it. All other graphs can be obtained by switching the roles of $a$, $b$ and $c$.
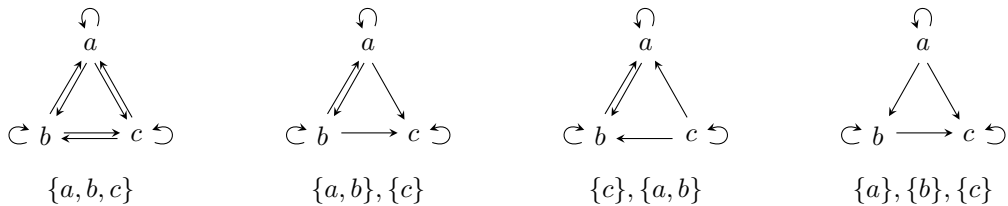
\begin{mathpar}
	\begin{tikzpicture}[auto, scale=1.2, font={\small},
	-{stealth[length=1mm,width=1mm]}, shorten <=1pt,
	{every loop/.style}={looseness=2, min distance=3mm}]
	\node (a) at (0,1) {$a$};
	\node (b) at (-0.5,0) {$b$};
	\node (c) at (0.5,0) {$c$};
	\path (a) edge[in=110,out=70,loop] (a)
		  (b) edge[in=160,out=200,loop] (b)
		  (c) edge[in=20,out=-20,loop] (c)
		  (a.260) edge (b.60)
		  (b.80) edge (a.240)
		  (b.10) edge (c.170)
		  (c.190) edge (b.-10)
		  (a.280) edge (c.120)
		  (c.100) edge (a.300);
    \node at (0,-0.6) {$\{a, b, c\}$};
	\end{tikzpicture}
\and
	\begin{tikzpicture}[auto, scale=1.2, font={\small},
	-{stealth[length=1mm,width=1mm]}, shorten <=1pt,
	{every loop/.style}={looseness=2, min distance=3mm}]
	\node (a) at (0,1) {$a$};
	\node (b) at (-0.5,0) {$b$};
	\node (c) at (0.5,0) {$c$};
	\path (a) edge[in=110,out=70,loop] (a)
		  (b) edge[in=160,out=200,loop] (b)
		  (c) edge[in=20,out=-20,loop] (c)
		  (a.260) edge (b.60)
		  (b.80) edge (a.240)
		  (b.10) edge (c.170)
		  (a.280) edge (c.120);
    \node at (0,-0.6) {$\{a, b\}, \{c\}$};
	\end{tikzpicture}
\and
	\begin{tikzpicture}[auto, scale=1.2, font={\small},
	-{stealth[length=1mm,width=1mm]}, shorten <=1pt,
	{every loop/.style}={looseness=2, min distance=3mm}]
	\node (a) at (0,1) {$a$};
	\node (b) at (-0.5,0) {$b$};
	\node (c) at (0.5,0) {$c$};
	\path (a) edge[in=110,out=70,loop] (a)
		  (b) edge[in=160,out=200,loop] (b)
		  (c) edge[in=20,out=-20,loop] (c)
		  (a.260) edge (b.60)
		  (b.80) edge (a.240)
		  (c.190) edge (b.-10)
		  (c.100) edge (a.300);
    \node at (0,-0.6) {$\{c\},\{a, b\}$};
	\end{tikzpicture}
\and
	\begin{tikzpicture}[auto, scale=1.2, font={\small},
	-{stealth[length=1mm,width=1mm]}, shorten <=1pt,
	{every loop/.style}={looseness=2, min distance=3mm}]
	\node (a) at (0,1) {$a$};
	\node (b) at (-0.5,0) {$b$};
	\node (c) at (0.5,0) {$c$};
	\path (a) edge[in=110,out=70,loop] (a)
		  (b) edge[in=160,out=200,loop] (b)
		  (c) edge[in=20,out=-20,loop] (c)
		  (a.260) edge (b.60)
		  (b.10) edge (c.170)
		  (a.280) edge (c.120);
    \node at (0,-0.6) {$\{a\},\{b\},\{c\}$};
	\end{tikzpicture}
\end{mathpar}

Note that this model does not allow crash failures; so all agents are active in all communication graphs.
It can be checked that the protocol complex functor associated with this dynamic network model is the same one as the functor given in Example~\ref{ex:async}.
Indeed, the four conditions on the sets $V_0, \ldots, V_m$ that appear in Equation~(\ref{eq:async}) characterize exactly the set of views associated with an ordered partition, as well known from \cite{herlihy}.
\end{example}

\begin{example}[Synchronous broadcast with crash failures]
\label{ex:broadcast-comm-pattern}
Let us 
consider now agents communicating though asynchronous message-passing. In this model, each agent $P_i$
(among a set of agents $P_k$, $k=0,\cdots,n$)
does the following, at each round,
\begin{itemize}
\item it broadcasts the value contained in its output buffer, $x^i_i$
to all other agents, by a succession of sends, in no definite order,
\item it stores data that has reached it from another agent, into its input buffer. The input buffer
is composed of locations $x^i_j$, $j=1,\cdots,\hat{i},\cdots,n$ in which
$P_i$ stores the value it received in the same round from $P_j$.
Notice that the input and output buffer form together an array of
$n+1$ locations, $x^i_k$, $k=0,\cdots,n$.
\item it then performs some local computation, basically amounting
to the calculation of some partial recursive function $f$ on
the values of some local variables $y_1,\cdots,y_l$ and the
assignment of the result to a local variable $x$. 
\end{itemize}
The full-information complex corresponding to this model is again the one where the local computation part is replaced by the history of communication. This comes directly from the formalization of this distributed architecture using communication patterns. 

Consider for instance the set of agents $A = \{a, b, c\}$ and the following communication graphs on~$A$:
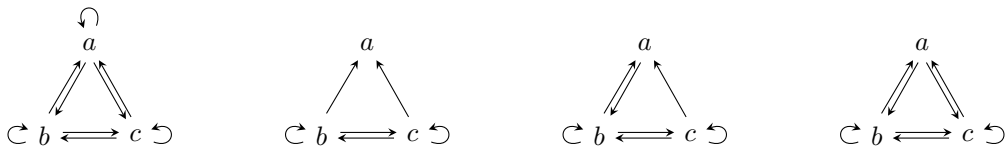
\begin{mathpar}
	\begin{tikzpicture}[auto, scale=1.2, font={\small},
	-{stealth[length=1mm,width=1mm]}, shorten <=1pt,
	{every loop/.style}={looseness=2, min distance=3mm}]
	\node (a) at (0,1) {$a$};
	\node (b) at (-0.5,0) {$b$};
	\node (c) at (0.5,0) {$c$};
	\path (a) edge[in=110,out=70,loop] (a)
		  (b) edge[in=160,out=200,loop] (b)
		  (c) edge[in=20,out=-20,loop] (c)
		  (a.260) edge (b.60)
		  (b.80) edge (a.240)
		  (b.10) edge (c.170)
		  (c.190) edge (b.-10)
		  (a.280) edge (c.120)
		  (c.100) edge (a.300);
	\end{tikzpicture}
\and
	\begin{tikzpicture}[auto, scale=1.2, font={\small},
	-{stealth[length=1mm,width=1mm]}, shorten <=1pt,
	{every loop/.style}={looseness=2, min distance=3mm}]
	\node (a) at (0,1) {$a$};
	\node (b) at (-0.5,0) {$b$};
	\node (c) at (0.5,0) {$c$};
	\path (b) edge[in=160,out=200,loop] (b)
		  (c) edge[in=20,out=-20,loop] (c)
		  (b.80) edge (a.240)
		  (b.10) edge (c.170)
		  (c.190) edge (b.-10)
		  (c.100) edge (a.300);
	\end{tikzpicture}
\and
	\begin{tikzpicture}[auto, scale=1.2, font={\small},
	-{stealth[length=1mm,width=1mm]}, shorten <=1pt,
	{every loop/.style}={looseness=2, min distance=3mm}]
	\node (a) at (0,1) {$a$};
	\node (b) at (-0.5,0) {$b$};
	\node (c) at (0.5,0) {$c$};
	\path (b) edge[in=160,out=200,loop] (b)
		  (c) edge[in=20,out=-20,loop] (c)
		  (a.260) edge (b.60)
		  (b.80) edge (a.240)
		  (b.10) edge (c.170)
		  (c.190) edge (b.-10)
		  (c.100) edge (a.300);
	\end{tikzpicture}
\and
	\begin{tikzpicture}[auto, scale=1.2, font={\small},
	-{stealth[length=1mm,width=1mm]}, shorten <=1pt,
	{every loop/.style}={looseness=2, min distance=3mm}]
	\node (a) at (0,1) {$a$};
	\node (b) at (-0.5,0) {$b$};
	\node (c) at (0.5,0) {$c$};
	\path (b) edge[in=160,out=200,loop] (b)
		  (c) edge[in=20,out=-20,loop] (c)
		  (a.260) edge (b.60)
		  (b.80) edge (a.240)
		  (b.10) edge (c.170)
		  (c.190) edge (b.-10)
		  (a.280) edge (c.120)
		  (c.100) edge (a.300);
	\end{tikzpicture}
\end{mathpar}
We name these graphs $G_1, \ldots, G_4$, from left to right.
Note that we omitted some graphs that can be obtained from those by permuting the names of the agents (i.e., graphs where agent~$b$ or~$c$ crashed instead of~$a$).
Intuitively,
\begin{itemize}
\item $G_1$ is an execution where no crash occurred, all messages were successfully delivered;
\item $G_2$, $G_3$, $G_4$ are executions where only agent~$a$ crashed, after sending $0$, $1$ or $2$ messages.
\end{itemize}
Among those communication graphs, only $G_1$, $G_2$ and $G_3$ have ``detectable crashes'', in the sense that whenever a agent is crashed, at least one of the remaining agents knows about it (because no message was received from the crashed agent).
So let us define two communication patterns: $P_{\text{detectable}}$ contains $G_1,G_2,G_3$ as well as graphs obtained from them by permuting the names of the agents (totalling $10$ graphs);
and $P_{\text{undetectable}}$ contains $G_1,G_2,G_3,G_4$ as well as permutations of them (totalling $13$ graphs).


Suppose we start with an initial chromatic semi-simplicial set $\C'$ with two triangles
~$w$ and~$w'$ that are glued along their $bc$-coloured edge.
In the case of undetectable crashes, we get the simplicial complex depicted on the right, with $19$ worlds named $w_0, \ldots, w_9$ and $w'_0, \ldots, w'_9$ (notice that $w'_5$ is missing).
%
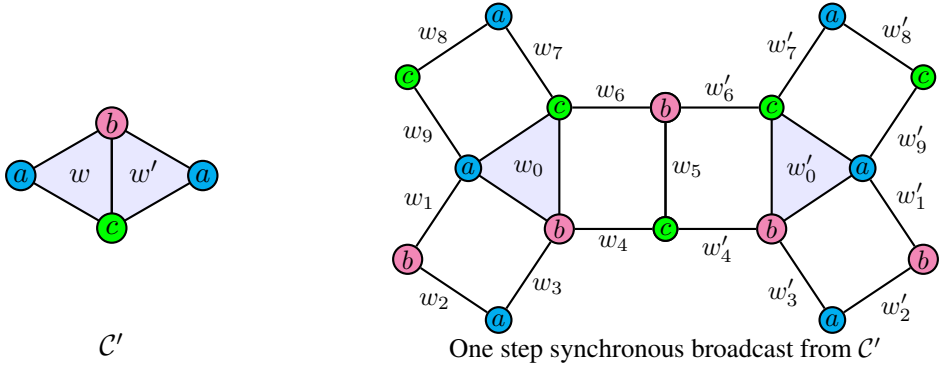
\begin{center}
	\begin{tikzpicture}[auto,scale=1.2,baseline={(0,0)},
	cloudgrey/.style={draw=black,thick,circle,fill=cyan,inner sep=1pt,minimum size=11pt},cloud/.style={draw=black,thick,circle,fill=magenta!60,inner sep=1pt,minimum size=11pt}, cloudblack/.style={draw=black,thick,circle,fill=green,inner sep=1pt,minimum size=11pt}]
	\draw[thick, draw=black, fill=blue!60, fill opacity=0.15]
	  (4,0) -- (5,-0.577) -- (5,0.577) -- cycle;
	\draw[thick, draw=black, fill=blue!60, fill opacity=0.15]
	  (6,0) -- (5,-0.577) -- (5,0.577) -- cycle;
	\node[cloudgrey] (b1) at (4,0) {$a$};
	\node[cloudgrey] (b2) at (6,0) {$a$};
	\node[cloudblack] (g1) at (5,-0.577) {$c$};
	\node[cloud] (w1) at (5,0.577) {$b$};
	\node at (4.65,0) {$w$};
	\node at (5.35,0.05) {$w'$};
	\node at (5,-1.85) {$\C'$};
	\end{tikzpicture}
	\hspace{2cm}
	\begin{tikzpicture}[scale=1.6, auto, font={\small}, rotate=90]
	\begin{pgfonlayer}{nodelayer}
		\node [style=cyan vertex] (0) at (0, 1) {$a$};
		\node [style=green vertex] (1) at (1, -0.5) {$c$};
		\node [style=magenta vertex] (2) at (-1, -0.5) {$b$};
		\node [style=green vertex] (3) at (1.5, 2) {$c$};
		\node [style=cyan vertex] (4) at (2.5, 0.5) {$a$};
		\node [style=magenta vertex] (5) at (1, -2.25) {$b$};
		\node [style=green vertex] (6) at (-1, -2.25) {$c$};
		\node [style=magenta vertex] (7) at (-1.5, 2) {$b$};
		\node [style=cyan vertex] (8) at (-2.5, 0.5) {$a$};
		\node [style=none] (12) at (-1, -0.5) {};
		\node [style=none] (13) at (0, 1) {};
		\node [style=none] (14) at (1, -0.5) {};
		\node [style=none] (15) at (0, 0) {$w_0$};
		\node [style=cyan vertex] (16) at (0, -5.5) {$a$};
		\node [style=green vertex] (17) at (1, -4) {$c$};
		\node [style=magenta vertex] (18) at (-1, -4) {$b$};
		\node [style=green vertex] (19) at (1.5, -6.5) {$c$};
		\node [style=cyan vertex] (20) at (2.5, -5) {$a$};
		\node [style=magenta vertex] (21) at (1, -2.25) {$b$};
		\node [style=green vertex] (22) at (-1, -2.25) {$c$};
		\node [style=magenta vertex] (23) at (-1.5, -6.5) {$b$};
		\node [style=cyan vertex] (24) at (-2.5, -5) {$a$};
		\node [style=none] (25) at (-1, -4) {};
		\node [style=none] (26) at (0, -5.5) {};
		\node [style=none] (27) at (1, -4) {};
		\node [style=none] (28) at (0, -4.5) {$w'_0$};
		\node [style=none] (29) at (-3, -2.25) {One step synchronous broadcast from $\C'$}; 
	\end{pgfonlayer}
	\begin{pgfonlayer}{edgelayer}
		\draw [style=thick] (2) to node [inner sep=1pt] {$w_3$} (8);
		\draw [style=thick] (8) to node [inner sep=1pt] {$w_2$} (7);
		\draw [style=thick] (7) to node [inner sep=1pt] {$w_1$} (0);
		\draw [style=thick] (0) to node [inner sep=1pt] {$w_9$} (3);
		\draw [style=thick] (3) to node [inner sep=1pt] {$w_8$} (4);
		\draw [style=thick] (4) to node [inner sep=1pt] {$w_7$} (1);
		\draw [style=thick] (1) to node [inner sep=2pt] {$w_6$} (5);
		\draw [style=thick] (5) to node [inner sep=2pt] {$w_5$} (6);
		\draw [style=thick] (6) to node [inner sep=2pt] {$w_4$} (2);
		\draw [style=filled path] (13.center)
			 to (14.center)
			 to (12.center)
			 to cycle;
		\draw [style=thick] (18) to node [inner sep=1pt, swap] {$w'_3$} (24);
		\draw [style=thick] (24) to node [inner sep=1pt, swap] {$w'_2$} (23);
		\draw [style=thick] (23) to node [inner sep=1pt, swap] {$w'_1$} (16);
		\draw [style=thick] (16) to node [inner sep=1pt, swap] {$w'_9$} (19);
		\draw [style=thick] (19) to node [inner sep=1pt, swap] {$w'_8$} (20);
		\draw [style=thick] (20) to node [inner sep=1pt, swap] {$w'_7$} (17);
		\draw [style=thick] (17) to node [inner sep=2pt, swap] {$w'_6$} (21);
		\draw [style=thick] (21) to (22);
		\draw [style=thick] (22) to node [inner sep=2pt, swap] {$w'_4$} (18);
		\draw [style=filled path] (26.center)
			 to (27.center)
			 to (25.center)
			 to cycle;
	\end{pgfonlayer}
\end{tikzpicture}

\end{center}
%
For instance, the world labelled~$w_5$ corresponds to both the effect of graph $G_2$ on $w$ and $w'$ at the same time. 
Indeed, when the communication graph~$G_2$ occurs, $a$ has crashed and the two agents~$b$ and~$c$ exchange information. But neither~$b$ nor~$c$ is able to distinguish between the initial worlds~$w$ and~$w'$. So no matter whether we started in~$w$ or~$w'$, the two remaining agents end up with the same local state. 
%
\end{example}



\section{Protocol complex algebras and co-algebras} 
\label{sec:protfunctor}

\subsection{Algebras induced by protocols} 

There is a last ingredient which will prove useful in the discussion of temporal-epistemic logics. In full-information protocol complexes, local states of agents are obtained as a pair of local initial state and of the full history of communication. This means that there is a projection operator, to the first component, from $F_!(X)$ to $X$, for any cset $X$ as we see below for round-based oblivious protocols:




\begin{proposition}  
\label{prop:alg}
The full information
protocol complex for any pattern of communication 
naturally comes with an algebra $F(X)\rightarrow X$. 
\end{proposition}

\begin{proof}
Let $U \in \Gamma$, we prove first that $F$ seen as a functor from $\Gamma$ to $\Gamma^{\op}\rightarrow \Set$, is such that there is a map from $F(U)$ to $\Gamma[U]$, the cset which has generated by a unique $U$-simplex. The result will hold by taking colimits over the chromatic simplexes of any $X \in \Gamma^{\op} \Set$ with the Yoneda extension of $F$.  

For all $V \subseteq U$, we define the map $f_U$ of cset, such that $f_U(V)$, $V\subseteq U$ from $F(U)(V)$ to $\Gamma(V,U)$ is defined as follows. To each set $\{(v,\view{G}{v}) \mid v \in V\}$ for some $G \in M_{U'}$ ($U'\subseteq U$ and $V \subseteq \activeProcs{G}$) we associate indeed the only $V$ simplex of $\Gamma[U]$, which is the inclusion of $V$ into $U$. It is easy to check that $f$ defines a morphism in $\Gamma^{\op}\Set$ since $f_U(V)$ is a sub-chromatic simplex of $f_U(W)$ in $\Gamma[U]$, for $V \subseteq W \subseteq U$. 
\end{proof}

\begin{example}[Synchronous broadcast case]
In the case of the synchronous broadcast of Example \ref{ex:broadcast-comm-pattern} from the simplicial set with two triangles, the corresponding algebra map is as follows. Triangles $w_0$ and $w'_0$ are mapped onto $w$ and $w'$ respectively. Call $w_{ab}$, $w_{bc}$, $w_{ac}$ (resp. $w'_{ab}$, $w'_{bc}$, $w'_{ac}$) the $ab$, $bc$, $ca$ 1-simplex boundary of $w_0$ (resp. of $w'$), the algebra map associates $w_1$, $w_2$, $w_3$ (resp. $w'_1$, $w'_2$, $w'_3$) on $w_{a,b}$ (resp on $w'_{a,b}$); $w_7$, $w_8$, $w_9$ (resp. $w'_7$, $w'_8$, $w'_9$) on $w_{a,c}$ (resp. $w'_{a,c}$); $w_4$, $w_5$, $w_6$ (resp. $w'_4$, $w'_5$, $w'_6$) on $w_{b,c}$ (resp. $w'_{b,c}$). 

\end{example}

As we can see in the example above, this natural algebra map sends states in the protocol complex from some input cset $I$ (i.e. states after one round of communication) to the states they originated from within $I$. 








\subsection{The co-algebraic view} 
\label{sec:coalgebras}

Co-algebras are the dual of algebras, as shown in the definition below. 
Modal logics are known to be related to $\Fun{G}$-co-algebras in the literature, see e.g. \cite{kurz}. 

\begin{definition}[\(G\)-coalgebra]
Let $G$ be an endofunctor, $G: \mathcal{D} \rightarrow \mathcal{D}$. A $G$-coalgebra is a pair $(D, c: D\rightarrow G(D))$ where $D \in \mathcal{D}$. Let $(D,c)$ and $(D',c')$ be two $G$-coalgebras. Then $g$ is a morphism of $G$-algebras from $(D,c)$ to $(D',c')$ if and only $g$ is a morphism from $D$ to $D'$ in $\mathcal{D}$ and the following diagram commutes: 
\begin{center}
\begin{tikzcd}[row sep=large]
D \arrow[d, "g"] \arrow[r, "c"]  &  G(D) \arrow[d, "G(g)"]     \\
D'                 \arrow[r, "c'"] &  G(D')                    \\
\end{tikzcd}
\end{center}
We write $coAlg(G)$ for the category of $G$-coalgebras. 
\end{definition}

A {\em non-deterministic} transition system naturally gives rise to an co-algebra over a certain (fixed) endofunctor on sets. 
Similarly to Section \ref{sec:algebras}, given a set $Ac$ of ``actions'' (or labels), consider an endofunctor $\Fun{F}$ on sets 
$X \mapsto \wp(X\times Ac)$. Then, a (set-theoretic) map $A: X \rightarrow \wp(X\times Ac)$ can be seen as a non-deterministic automaton: indeed, given $x \in X$, where $x$ is seen as a state, the function $A$ assigns to $x$ a set of pairs $(y,a)$ made up of an action $a$ and a state $y$, to be interpreted as the states $y$ that can be reached from $a$ by firing action $a$.

From a given $F$-algebra, there is a simple, and bi-univoque way to generate a $G$-algebra, when $G$ is right-adjoint to $F$ as we see below: 

\begin{lemma}
\label{lem:adjointcoalg}
Suppose that $F: \mathcal{C} \rightarrow \mathcal{C}$ is left adjoint to functor $G: \mathcal{C} \rightarrow \mathcal{C}$. Then the category of $F$-algebras is isomorphic to the category of $G$-coalgebras. 
\end{lemma}

\begin{proof}
Take $(C,a: F(C)\rightarrow C)$ a $F$-algebra and consider the composite of $G(a)$ with the unit $\epsilon_C$ on $C$ of the adjunction: 
\begin{center}
\begin{tikzcd}[row sep=large]
C \arrow[r, "\epsilon_C"]  &  GF(C) \arrow[r, "G(a)"] & G(C)    
\end{tikzcd}
\end{center}
This defines a $G$-coalgebra. Inversely, take $(D,c: D\rightarrow G(D))$ a $G$-coalgebra and consider the composite of $F(c)$ with the counit $\epsilon_C$ on $C$ of the adjunction: 
\begin{center}
\begin{tikzcd}[row sep=large]
F(D) \arrow[r, "F(c)"]  & FG(D) \arrow[r, "\eta_D"] & D    
\end{tikzcd}
\end{center}
\noindent is a $F$-algebra on $D$. 
The fact that these two operations are inverse from one another is a direct consequence of the axioms on units and counits of adjunctions. 
\end{proof}

But indeed, the $G$-algebra corresponding through Lemma \ref{lem:adjointcoalg} to the $F$-algebra of deterministic automata, Section \ref{sec:algebras} is not the one, of non-deterministic automata, that we mentionned in the beginning of this Section. We recall that deterministic automata are $F$-algebras for functor $F$ associating to each set $X$, the set $X\times Ac$. The right-adjoint to $F$ is $G$ associating to each set $X$, the set of functions from $Ac$ to $X$. Indeed $G$-algebras are just another encoding of deterministic automata. A map from $X$ to $G(X)$ is just a map which associates to each state $x$ and action $a$, a ``next'' state $y$. 

In the situation of full-information protocol complex functors, we have naturally a right-adjoint $F^*$ to the protocol complex functor $F_{!}$. So $F_{!}$-algebras correspond to $F^*$-coalgebras. In some sense, the $F_!$-algebras we have are really similar to the encoding of deterministic transition systems, but in theory $F^*$-coalgebras should be more ``powerful'' and also could deal with even non-bounded non-deterministic transition systems. 

This co-algebraic view of protocol complex functors are in relation with the classical ``carrier map'' view of protocols, as we show below in a particular case. 

By Lemma \ref{fundamental}, the $G$-coalgebras that correspond to our $F$-algebras implementing protocol complexes, are defined by $G(X)=\Gamma^{\op}\Set(F(.),X)$ meaning that the $\gamma$ simplices of $G(X)$, for $\gamma \in \Gamma$,  $G(X)(\gamma)$ is the set of maps from $F(\gamma)$ to $X$. Now, a $G$-coalgebra is a map in $\Gamma^{\op}\Set$ from $X$ to $G(X)$. This maps to any $\gamma$-simplex of $X$, a map from $F(\gamma)$ to $X$. This map from $F(\gamma)$ to $X$ corresponds to the carrier map applied to the set of participating processing $\gamma \in \Gamma$ of \cite{herlihy}. 








\section{Iterated protocol complexes and free algebras} 
\label{sec:iterated}

We now note that:

\begin{lemma}
\label{lem:mono}
Let $F: \ \Gamma \rightarrow \Gamma^{\op} \Set$ be a functor that sends inclusion maps $A \hookrightarrow B$ of elements $A$, $B$ of $\Gamma$ to inclusions of csets $F(A) \hookrightarrow F(B)$. Then its left Yoneda extension $F_!$ preserves monomorphisms. 
\end{lemma}

\begin{proof}
As well known, \cite{maclane1992sheaves} (and \cite{Karazeris} for the more general case), left Kan extensions $F_!$ of a functor $F: \ \mathcal{C}\rightarrow \mathcal{E}$ preserve all colimits as left adjoints, and any finite limit when $F$ preserves these finite limits, at least in the case when $\mathcal{E}$ is a topos. This is the case here where $\mathcal{C}=\mathcal{E}=\Gamma^{\op}\Set$. 
Indeed, as well-known as well, monomorphims are morphisms obeying a certain pullback diagram, hence if $F$ preserves monomorphisms, $F_{!}$ preserves monomorphisms as well. Finally, monomorphisms in $\Gamma$ and in $\Gamma^{\op} \Set$ are inclusion maps. 
\end{proof}

It is clear that for distributed applications, $F(B)$ is going to contain as a subcomplex $F(A)$ where $A$ is a subset of agents of $B$. We particularize this to well-known protocol complex examples for a start:

\begin{example}
We have already noted in Example 
\ref{ex:async}, $G$ corresponding 
to 
the immediate snapshot model, preserve inclusion maps $\delta_{U,V}$. By Lemma \ref{lem:mono}, 
$G_!$ preserve monomorphisms. As it is a left adjoint, it also preserves colimits. 
\end{example}

For general patterns of communication, we have the same phenomenon:

\begin{lemma}
\label{lem:preserve}
The protocol complex functors for general patterns of communication (Section~\ref{sec:protocolfunctor})
defines a functor $F_!: \ \Gamma^{\op} \Set \rightarrow \Gamma^{\op} \Set$ that preserves colimits and monomorphisms.
\end{lemma}

\begin{proof}
As protocol complexes defined in Section \ref{sec:protocolfunctor} are defined by Yoneda extensions, they preserve colimits, as they are left-adjoints, see Lemma \ref{basic}. 

Now, we observe that the protocol complex functors defined in Section \ref{sec:protocolfunctor} are such that $F(\delta_{U,T})$ is a monomorphism, as this defines an inclusion of the set of $V$-simplexes of $F(U)$ into the set of $V$-simplexes of $F(T)$. The result is then a consequence of Lemma \ref{lem:mono}.
\end{proof}

The dynamics of a distributed system in the framework of e.g. \cite{herlihy} is given by the iteration of the protocol complex functor. The natural question is how to wrap up formally all the iterations of the protocol complex. The right formalization is given through the notion of free algebra as we know.
Now, thanks to Theorem \ref{thm:thm1}, Lemma \ref{lem:mono} and Lemma \ref{lem:preserve}, we know that in our general case where our protocol complex satisfies Hypotheses (1) to (4), we know that they provide us with varietors, and that we have a general Kleene-like formula for the free algebra generated by a given cset.


\section{Temporal-epistemic logics on free algebras} 
\label{sec:temporalepistemic}

In this section, we develop a temporal-epistemic logics which has a natural semantics in terms of chromatic simplicial sets, and which will be adapted to specifying and verifying the diverse protocols we have been modelling in the previous sections. 

\subsection{The logics} 
\label{sec:logics}

In this section, we are going to consider the following modal logics, extending first-order predicate logics on a given set $\AP$ of atomic predicates: 
$$\phi ::= p \ \mid \ \neg p \ \mid \ \phi\vee \phi \ \mid \ K_a \varphi \ \mid \ C_B \varphi \ \mid D_B \varphi \ \mid \ X \varphi \ \mid \ \Box \varphi \ \mid \ \Diamond \varphi $$
\noindent where $p \in \AP$, 
$K_a \varphi$ reads ``agent $a$ knows $\varphi$'', $C_B \varphi$ means ``subgroup $B$ of agents have the common knowledge of $\varphi$'', $D_B \varphi$ means ``subgroup $B$ of agents have distributed knowledge $\varphi$'', $X \varphi$ read ``next $\varphi$'', resp. $\Box \varphi$, ``always $\varphi$'' and $\Diamond \varphi$, ``eventually $\varphi$''.


We consider also the following predicates, that will reveal convenient in the sequel:

\begin{align*}
& \deadprop{a} \,:=\, K_{a} \false
& &
\aliveprop{a} \,:=\, \widehat{K}_{a} \true
\\
& \deadprop{U} \,:=\, \bigwedge_{a \in U} \deadprop{a}
& &
\aliveprop{U} \,:=\, \widehat{D}_U \true
\end{align*}

As already used in \cite{STACS22}, $\deadprop{a}$ precisely models the fact that an agent is crashed in some world, as being equivalent to the fact it knows ``$\false$''. Properties $\aliveprop{a}$, $\deadprop{U}$ and $\aliveprop{U}$ are all derived from this. 



\subsection{Chromatic augmented semi-simplicial sets as simplicial models}  


In this section, we give a semantics of the logics of Section \ref{sec:logics} based on the ``maximal epistemic covering model'' of \cite{LICS2023} (also similar to the maximal simplicial complex models of \cite{arxivfaulty}). In these maximal models, all simplexes, of any dimension, are worlds in the sense of the more classical Kripke semantics. We also suppose we have a locality axiom, that is, $\AP$ is partitionned into $\AP_a$, for $a\in A$, sets of ``local atomic predicates'' for agent $a$.


\begin{definition}[Semi-simplicial model]
\label{def:semimodel}
A \emph{semi-simplicial model} $X = \langle S, \ell \rangle$ consists of a chromatic semi-simplicial set $S$, together with a labelling $\ell: S_0\to \mathscr{P}(\AP)$ that associates with each vertex $s$ of $S$ a set of atomic propositions~$\ell(s)\subseteq \AP_{\chi(s)}$ that hold there. A morphism of semi-simplicial models $\alpha : X \to Y$ is a morphism of chromatic semi-simplicial models that does not create new labelling: $ \ell'(\alpha(s)) \subseteq \ell(s)$.
We denote by $cset_\AP$ the category of semi-simplicial models.

Note that for any semi-simplicial model $\la S, \ell \ra$, the labelling function $\ell$ can be extended to a function, still written $\ell$ from all simplexes of $S$ to $\mathscr{P}(\AP)$ by, for $w \in S_n$ a $n$-simplex ($n\geq 0$),  \begin{equation}
\label{eq:ell}
    \ell(w)=\bigcup\limits_{0\leq i_1 < i_2 < \ldots < i_n \leq n} \{\ell(d_{i_1}\circ d_{i_2}\circ \ldots \circ d_{i_n}(w))\}
\end{equation}
\end{definition}


Indeed, the notions of algebra, coalgebra and free algebra lift to category $cset_\AP$. 

It is also straightforward to extend the cset description of the protocol complex functor for all patterns of communication to a functor from $cset_\AP$ to $cset_\AP$. Given a simplex $s$ of a $cset_\AP$ $X=\langle S,\ell\rangle$ with vertices $s_0, \ldots,s_m$, we produce the cset $F(s)$ for a given pattern of communication, and a $cset_\AP$ that we denote also by $F(s)$, where all vertices $v$ of $F(s)$ have associated predicate $\ell(s_{\chi(v)})$. This means that we just copy the predicates defining the ``initial value'' of the corresponding agent, to the vertices in the image of $F$.

\subsection{Semantics in free algebras over chromatic simplicial sets} 
\label{sec:semantics-cset}


The semantics of the temporal-epistemic logics we are considering can be given in any object $I^\infty$ of the free $\Fun{F}$-algebra over some element $I\in cset_\AP$, given an algebra $q: \ \Fun{F}(I) \rightarrow I$, as we will see now. 

We note $lnk_a(x,y)$ (resp. $lnk_B(x,y)$), for $x, y \in C$, $C$ a given cset, for the following predicate. We have $lnk_a(x,y)$ to be true iff there is an inclusion map $i$ from $\Gamma^{op}\Set[a]$ (resp. $\Gamma^{op}\Set[B]$) to $C$ with the image of $i$ included in both $x$ and $y$. Geometrically speaking, $lnk_a(x,y)$ is true iff $x$ and $y$ share a $a$ (resp. $B$) simplex in $C$. 

Our semantics takes the following form: 

\begin{tabular}{llrl}
(At) & $I^\infty,x \models p$ & iff & $p \in \ell(x)$\\
(Neg) & $I^\infty,x \models \neg \varphi$ & iff & $I^\infty,x \not\models \varphi$\\
(And) & $I^\infty,x \models \varphi \wedge \psi$ & iff & $I^\infty,x \models \varphi
\mbox{ and } I^\infty,x \models \psi$\\
(K) & $I^\infty,x \models K_a\, \varphi$ & iff & $\mbox{for all } y \in I^\infty, lnk_a(x,y) \mbox{ implies } I^\infty,y \models \varphi$\\
(C) & $I^\infty,x \models C_B\, \varphi$ & iff & $\mbox{for all } y \in  I^\infty, y=y_n, \ lnk_B(y_i \cap y_{i-1}), \ i=1,\ldots,n, $\\ 
& & & $ \ \ \ \ \ \ \ y_1=x \mbox{ implies } I^\infty,y \models \varphi$\\
(D) & $I^\infty,x \models D_B\, \varphi$ & iff & $\mbox{for all } y \in I^\infty, lnk_B(x,y) \mbox{ implies } I^\infty,y \models \varphi$\\
(X) & $I^\infty,x \models X\, \varphi$ & {iff} & $\mbox{for all } y \in \Fun{F}^{n+1}(I), \mbox{ $n$ s.t. $x\in \Fun{F}^n(I)$ }, \Fun{F}^n(q)(y)=x$ \\
& & & \ \ \ \ \ $\mbox{ implies } I^\infty,y \models \varphi$ \\
(Box) & $I^\infty,x \models \Box\, \varphi$ & {iff} & $\mbox{for all } i \geq 0,  
I^\infty,x \models X^i \varphi$ \\
(Diam) & $I^\infty,x \models \Diamond \, \varphi$ & {iff} & 
$\mbox{exists } i \geq 0,  
I^\infty,x \models X^i \varphi$ \\
\end{tabular}

 This reads as follows: 
\begin{itemize}
\item (At), (Neg), (And) give the obvious semantics of the underlying first-order logic
\item (K) is the interpretation of the knowledge modal operator of \cite{LICS2023}: agent $a$ knows $\varphi$ in world $x$ if $\varphi$ is true in all worlds $y$ that are connected to world $x$ through vertex $a$ of $x$.
\item (C) and (D) are similar to \cite{LICS2023} as well: the group $B$ of agents has common knowledge of $\varphi$ in world $x$ if $\varphi$ is true on all worlds $y$ that are reachable from a path of worlds that are connected to $x$ by vertices labelled by $B$. And $B$ has distributed knowledge in world $x$ of $\varphi$ if $\varphi$ is true in all worlds $y$ that are connected to $x$ through vertices in $B$ of $x$. 
\item Equation (X) defines the semantics of the temporal modality ``next step'' and reads as follows: $\varphi$ is true after one step from world $x$, which is in $I^\infty$ hence in $\Fun{F}^n(I)$ for some $n\in \N$, which is necessarily unique (since $I^\infty$ is the infinite coproduct of all the $\Fun{F}^n(I)$, $n\in \N$, by Theorem \ref{thm:thm1}) if $\varphi$ is true in all worlds $y$ in the chromatic simplicial set in the next step $\Fun{F}^{n+1}(I)$, which are possible next steps of $x$. The interpretation of the ``always'' (Box) and ``eventually'' (Diam) temporal modalities are then obvious. 
\end{itemize}



Note that simplicial models generate history LTS in the sense of \cite{Knight}. The latter structures are defined as the unfolding of transition systems together with an indistinguishability (equivalence) relations on states (one per agent), in such a way that the equivalence naturally lifts to histories. Indeed, our free algebra model is a history based model, where the states are all simplexes, and the indistinguishability relation for an agent $a$ between two states $x$ and $y$ is given by $lnk_a(x\cap y)$. The gain we are making with our modelling is, apart from a clean categorical construction, the use of geometric properties of the state space (modeled by simplicial sets), for proving or disproving the existence of tasks that verify some specification. 


Note also that in the particular case of the immediate snapshot model (Example \ref{ex:async}), the Dynamic Epistemic Logic (DEL) framework of \cite{gandalf} is a fragment of the logic we give in this section: rule (X) given for ``next $\varphi$'' indeed reflects the effect of the update or public announcement that is made in the corresponding action model, on the validity of formula $\varphi$. 




\subsection{Axiomatics} 

First, we note that we have an analogous of the ``knowledge gain'' theorem of \cite{infcomp}. This will allow us in particular to derive some sound axioms that relate the temporal modalities with the epistemic modalities.  

An epistemic logic formula $\phi \in \mathcal{L}_K$ is called \emph{positive} when it does not contain negations, except possibly in front of atomic propositions.
Formally, positive formulas are built according to the following grammar:
$$
\varphi ::= p \mid \neg p \mid \varphi \land \varphi \mid \varphi \lor \varphi \mid 
K_a\varphi \mid C_B \varphi \mid D_B \varphi \qquad a \in A,\; p \in \AP
$$
We write $\PL{K}$ for the set of positive epistemic formulas.
Essentially, positive formulas forbid talking about what an agent does not know.

\begin{theorem}[knowledge gain]
\label{thm:lose-knowledge}
Consider simplicial models $\C=\la S,\ell \ra$ 
and
$\D=\la S',\ell'\ra$, 
and a  morphism of pointed simplicial models $f : (\C,w) \to (\D,w')$. 
Let $\varphi \in \PL{K}$ be a positive formula. 
Then $\D,w' \models \varphi$ implies $\C,w \models \varphi$.
\end{theorem}
\begin{proof}
We proceed by induction on the structure of the (positive) formula $\varphi$.
For the base case $\varphi = p \in \AP$, 
as $\D, w' \models p$ we know that $p\in \ell'(w')$. But as $f$ is a morphism of semi-simplicial models and $f(w)=w'$, this implies that $\ell'(w')=\ell'(f(w))\subseteq \ell(w)$. Thus, $p\in \ell(w)$ and $\D, w \models \varphi$.

A similar reasoning works in the case of a negated atomic proposition $\neg p$.
The cases of conjunction and disjunction follow trivially from the induction hypothesis.

Suppose now that $\D,w' \models K_a \varphi$. 
There are two cases: 
there is no $a$-colored vertex in $w'$, therefore, as $w'=f(w)$, there is no $a$-colored vertex in $w$ either, 
therefore it also trivially holds that $\D,w \models K_a \varphi$. 
The other case is that there is an $a$-colored vertex in $w'$, and $w$. 
Suppose now that $lnk_a(w,z)$ for some facet $z$ of $\C$, and let us prove that $\C, z \models \varphi$. Let $v$ be an $a$-coloured vertex common to $z$ and $w$. Then $f(v)$ is an $a$-coloured vertex in both $f(z)$ and $f(w)$.
And because $\D, w'\models K_a \varphi$, $\D,f(z)\models \varphi$ and by induction on the formula, $\C,z \models \varphi$. 

The rest follows easily for the other connectives. 
\end{proof}

This now has the following consequence on the interpretation of our temporal epistemic logics as defined in Section \ref{sec:semantics}, for free algebras over chromatic simplicial sets: 

\begin{lemma}
For any algebra $q: \Fun{F}(Q) \rightarrow Q$ and $\varphi$ any positive epistemic formula, we have, for $x \in \Fun{F}^n(Q)$: 
$$
Q^\infty, x \models \varphi \ \Rightarrow \ Q^\infty, x \models X \varphi 
$$
\noindent and similarly, for $x \in \Fun{F}^n(Q)$: 
$$
Q^\infty, x \models \varphi \ \Rightarrow \ Q^\infty, x \models \Box \varphi 
$$
$$
Q^\infty, x \models \varphi \ \Rightarrow \ Q^\infty, x \models \Diamond  \varphi 
$$
\end{lemma}

\begin{proof}
For $\varphi$ a positive formula, as $q: \Fun{F}(Q)\rightarrow Q$ is a morphism of simplicial models, we know by Theorem \ref{thm:lose-knowledge} that $Q, y \models \varphi$ implies $\Fun{F}(Q), x \models \varphi$ when $q(y)=x$. 
Similarly, $\Fun{F}^n(q): \ \Fun{F}^{n+1}(Q) \rightarrow \Fun{F}^n(Q)$ is a morphism of simplicial models, so it holds as well that
$\Fun{F}^n(Q), x \models \varphi$ implies $\Fun{F}^{n+1}(Q), y \models \varphi$ when $\Fun{F}^n(q)(y)=x$, i.e. $Q^\infty, x \models X \varphi$.
By the semantics of $X \varphi$ this proves the first result.

For the second part, this is done by induction on the number of applications of $q$. Finally, we observe that $\Box \varphi$ implies $\diamond \varphi$. 
\end{proof}

In particular, this means that our distributed machines never ``forget'' a true statement about the initial values of agents, that they may have learned through communication: they can only learn more, through communication. 


On top of the axioms for normal modal logic, 
with all propositional tautologies,
closure by modus ponens, and the necessitation rule, get the following set of sound axioms for our temporal-epistemic logics: 

\begin{lemma}
The following axioms are sound in all free algebras over chromatic simplicial sets: 
\begin{itemize}
    \item ($\mathbf{K}$) $D_U (\varphi \Rightarrow \psi) \Rightarrow (D_U\varphi \Rightarrow D_U \psi)$
    \item ($\mathbf{4}$) $D_U\varphi \Rightarrow D_UD_U\varphi$
    \item ($\mathbf{B}$) $\varphi \Rightarrow D_U\neg D_U\neg\varphi$
    \item ($\mathbf{Mono}$) for $U\subseteq U'$, $D_U\varphi \Rightarrow D_{U'}\varphi$
    \item ($\mathbf{Union}$) for $U,U'$, $\aliveprop{U}\land \aliveprop{U'} \Rightarrow \aliveprop{U\cup U'}$
    \item \makebox[1.1cm]{($\mathbf{NE}$) \hfill} $\bigvee_{a\in A} \aliveprop{a}$;
    \item \makebox[1.1cm]{($\mathbf{P}$)\hfill} $\aliveprop{U} \land \deadprop{U^c} \land \varphi \Rightarrow D_U(\deadprop{U^c}\Rightarrow \varphi)$;
    \item \makebox[1.1cm]{($\mathbf{Max}$)\hfill} for $U\not=\emptyset$, $\aliveprop{U} \Rightarrow \neg D_U \neg \deadprop{U^c}$;
\item \makebox[1.1cm]{($\mathbf{GFP}$) \hfill}  $C_A(\varphi\implies \bigwedge\limits_{a\in A} K_a \varphi) \implies (\varphi \implies C_A \varphi)$
%
%

\item \makebox[1.1cm]{($\mathbf{KG}$)} 
for $\varphi$ positive epistemic formula: $\varphi \Rightarrow X \varphi$
\item \makebox[1.1cm]{($\mathbf{KG\Box}$)} 
for $\varphi$ positive epistemic formula: $\varphi \Rightarrow \Box \varphi$
\item \makebox[1.1cm]{($\mathbf{KG\Diamond}$)} 
for $\varphi$ positive epistemic formula: $\varphi \Rightarrow \Diamond \varphi$
\item \makebox[1.1cm]{($\mathbf{AX}$)}
$\Box \varphi \Rightarrow X \Box \varphi$
\item \makebox[1.1cm]{($\mathbf{XX}$)}
$\Box \varphi \Rightarrow \varphi$
\item \makebox[1.1cm]{($\mathbf{AI}$)}
$\Box(\varphi \Rightarrow \psi) \Rightarrow \Box \varphi \Rightarrow \Box \psi$
\item \makebox[1.1cm]{($\mathbf{E}$)} 
$\Diamond \varphi \Leftrightarrow \neg \Box \neg \varphi$
\end{itemize}
\end{lemma}

\begin{proof}
The first five axioms ($\mathbf{K}$), ($\mathbf{4}$), ($\mathbf{B}$), ($\mathbf{Mono}$) and ($\mathbf{Union}$) are called KB4$_n$ + Mono + Union in \cite{LICS2023} and are known to hold for the epistemic part of these csets. 

The three following axioms ($\mathbf{NE}$), ($\mathbf{P}$) and ($\mathbf{Max}$) hold because of Theorem 3 of \cite{LICS2023}: 
the logic $\ECn \mathbf{+ NE + P + Max}$
is sound with respect to all chromatic semi-simplicial models.

Because of Theorem \ref{thm:lose-knowledge}, we have also the following sound axioms: ($\mathbf{KG}$), 
($\mathbf{KG\Box}$) and 
($\mathbf{KG\Diamond}$). 

We also check easily the classical axioms for temporal logics: ($\mathbf{AX}$), 
($\mathbf{XX}$), 
($\mathbf{AI}$) and 
($\mathbf{E}$).
For instance, suppose that $I^\infty, x \models \Box \varphi$. This means that for all $i\geq 0$, $I^i,x\models \varphi$, so in particular, $I,x \models \varphi$. Thus, $I^\infty, x \models (\Box \varphi \implies \varphi)$, for all models $I$ and $x \in I$. 
\end{proof}



\subsection{Example: task specification in fault-tolerant distributed computing} 
\label{sec:applications}



We slightly adapt the framework for task specification of 
\cite{herlihy}, that we recapped in Section \ref{sec:simplicialcomplex}, to our framework, in which we use csets instead of chromatic simplicial complexes. 

A task $\mathcal{T}$ is given by a cset of inputs $I$, a cset of outputs $O$, and a task specification $T$, a sub-cset of $I\times O$, the product of these two csets defined as follows: $(I \times O)_n=I_n\times O_n$ with $d_U(i,o)=(d_U(i),d_U(o))$. This cset $T$ relates inputs with potential outputs, that any protocol solving the task should choose from. Naturally, $T$ comes with a projection morphism $\pi_T: \ T \rightarrow I$.

Given a particular architecture of communication (dynamic network in our case, or pattern of communication), formalized by a protocol complex $P$ (or an iterated protocol complex) and the map $\pi_P$ relating worlds in $P$ (simplexes, or states) to worlds in the input complex $I$, we define solvability of a task $\mathcal{T}$ on this particular architecture as follows. Mathematically, task ${\mathcal{T}}$ is \emph{solvable} using the protocol complex $P$ if there exists a cset morphism $\delta: P\rightarrow T$ such that $\pi_{I}\, \circ\, \delta=\pi_P$, i.e., the diagram of simplicial complexes below commutes.

\begin{center}
\begin{tikzpicture}
  \node (s) {$P$};
  \node (xy) [below=2 of s] {${T}$};
  \node (x) [left=of xy] {$I$};
  \draw[<-] (x) to node [sloped, above] {$\pi_P$} (s);
  \draw[->, right] (s) to node {$\delta$} (xy);
  \draw[->] (xy) to node [below] {$\pi_I$} (x);
\end{tikzpicture}
\end{center}

Our temporal-epistemic logics allows for more general task specifications in that we can specify the existence of such decision tasks at specific time instants (round numbers) during protocol execution.

For instance, for binary consensus, the atomic predicates in $\AP_i$ are equality predicates of their initial state with 0 or with 1 and the specification can be written in our temporal-epistemic logics as \[\Diamond \left(C_A(\bigvee\limits_{i\in A} p_i=0)\wedge C_A(\bigvee\limits_{i\in A} p_i=1)\right)\]

\section{Beyond oblivious patterns of communication and generic protocols} 

\label{sec:beyond}

In this section, we discuss another presheaf model for distributed models 
which encodes protocols that can evolve according to non only participating values, but also 
local values of agents. The evolution of the system may depend on local decision values, and not just on 
the patterns of communication. 

This model gives rise to a presheaf category, and the same machinery as the one used in the previous section can be used again, in particular, the definition of protocol functors through Yoneda extensions, algebras on these functors, and free algebras. We then discuss the strong connections that we may find between these approaches in Section \ref{sec:rel}.

Also, in the previous sections, we only considered ``generic protocols'' (full-information protocols in the case of e.g. the immediate snapshot model of Example \ref{ex:async}); in Section \ref{sec:interpretedprotocol}, we show that protocol functors can only be defined for particular ``interpreted protocols''. 

\subsection{Chromatic simplicial sets with decisions} 
\label{sec:decision}








Let $\Psi$ be the category \(\Gamma\) indexed by a category $\CatValues$ of decision values. The objects of \(\Psi\) are $(A,D)$ with $A\in \Gamma$ and $D: A \rightarrow \mathcal{V}\cup \{\bot\}$, a functor from the poset category \(A\) to \(\CatValues\). Morphisms $f: \ (A,D)\rightarrow (B,E)$ are morphisms from $(A,D)$ to $(B,E)$ in $\Psi$. 


\begin{definition}[Chromatic simplicial sets with decision values]
We call the presheaf category $\Psi^{\op}\Set$ the category of chromatic simplicial sets with decision values.
\end{definition}

Equivalently, by the fundamental theorem of topoi \cite{maclane1992sheaves}, objects of \(\Psi^{\op}\Set\) can be seen as csets together with a labeling map to a chromatic simplicial set of possible decision values $\mathcal{O}$, defined as follows.
Chromatic simplices $s$ (of dimension $m+1$) of $\mathcal{O}$ are $s=((a_0,d_0),\ldots,(a_{m},d_m))$ with $a_i \in \N$, $d_i \in \Values$, satistying $a_0 < \ldots < a_{m}$. 
The color of $s$ is $(a_0,\ldots,a_m)$ and the boundary operators are the obvious ones: $d_{a}(s)=((a_0,d_0),\ldots,\widehat{(a_i,d_i)},(a_{m},d_m))$ when $a_i=a$. \(\Psi^{\op}\Set\) is equivalent to \(\Gamma^{\op}\Set / \mathcal{O}\), for \(\Psi\) and \(\mathcal{O}\) defined with the same \(\CatValues\).

This is the category in which we can describe protocols where each step may depend on local decision values, or an internal state such as a position, in distributed robotics:

\begin{example}[Distributed multi-robot systems]
\label{ex:multirobot}
The task specification framework is also relevant for distributed systems that interact in the physical world, such as multi-robot systems. The distributed coordination of robots can be understood from the distributed computing perspective. The main model is that of luminous robots executing the look-compute-move (LCM) scheme \cite{LCM}, described next. Consider a set of robots \(R_i\), \(i = 1,\ldots,n\), with positions \(\mathbf{x}_i \in \R^p\) (for some \(p \geq 1\)). Each robot is equipped with a set of visible lights used for communication and memory. At each (global) step, they perform the following actions:
\begin{enumerate}
  \item (look) They acquire the position and lights of the robots \(R_j\), \(j \neq i\), that are within their visibility radius \(r\).
  \item (compute) They compute their next move and new configuration of lights from what they observed.
  \item (move) They update their new position and lights to the computed values.
\end{enumerate}

We suppose that the dynamics of each robot is such that at each step, robot $R_i$ at current position $\mathbf{x}_i$ can only move to positions within a set $N_i(\mathbf{x}_i)$. Several properties of distributed systems are immediately translated, such as communication faults, synchronicity and limited interaction range\footnote{Other difficulties are introduced, such as a visibility function depending on local states and positions that do not immediately update to the computed ones, as is often the case with physical systems. We will not consider those cases now.}. In many ways this looks similar to the atomic snapshot model of Example \ref{ex:async}, where the scan operation corresponds to ``look'', and update corresponds to ``compute-move''. This has been observed in \cite{LCM}, at least when the communication radius is infinite.

The protocol complex functor $F_R$ corresponding to this case is constructed as follows. We consider decision values to be $\Values=\R^p$, that will contain robots' individual positions. Then $F_R$ is going to be the Yoneda extension of the functor, that we still write $F_R$, which to any $(A,D)\in \Psi$ associates the chromatic simplicial set with decisions, with the set of $(B,E)$-simplexes defined as follows:
\begin{itemize}
\item $\emptyset$ if $B \not \subseteq A$ or if $E(b_i)\not \in N_i(x_i)$ for some $i=1,\ldots,n$. 
\item Or when $B=\{b_0,\ldots,b_m\}\subseteq A$ and $E(b_i)\in N_i(x_i)$ for all $i=1,\ldots,n$:
$$\left\{\{(b_0,W_0,E(b_0)),\cdots,(b_m,W_m,E(b_m))\} /
\left\{\begin{array}{l}
W_i \subseteq A \\
\forall i, \ 0 \leq i \leq m, \ b_i \in W_i \\
\forall i,j, \ 0 \leq i, j \leq m, \ (W_i \subseteq W_j) \mbox{ or } (W_j \subseteq W_i) \\
\forall i, j, \ 0 \leq i, j \leq m, \ b_i \in W_j \implies W_i \subseteq W_j \\
\forall i,j, \ d(E(b_i),E({b_j}))\leq r \\
\end{array}\right.\right\}$$
\end{itemize}


The first condition imposes that, among the states in the next step of execution from a $(A,D)$-simplex, there can only be $(B,E)$-simplexes whose set of participating robots $B$ is a subset of the participating robots before that step of execution, i.e. $A$. Furthermore, it requires that the positions of robots $R_i$ after that step, i.e., after one move operation, are within the allowed positions from $x_i$, i.e., within $N_i(x_i)$.

In the case that those two conditions are respected, the next step is composed of $(B,E)$ simplexes with all possible views that they may get from the concurrent execution of the look-compute-move steps, corresponding exactly to the immediate snapshot conditions of Example~\ref{ex:async}, except those that are forbidden because of the visibility constraints between each pair of participating robots in that state: $d(E(b_i),E(b_j))\leq r$.  

Similarly to Lemma \ref{lem:preserve}, $F_R$ preserves monomorphisms, and as a left adjoint ($F_R$ is defined as a Yoneda extension in a presheaf category), it preserves colimits as well. Also, similarly to Proposition \ref{prop:alg}, we have a natural algebra $F_R(X)\rightarrow X$ for all $X \in \Psi^{\op}\Set$, which has as underlying algebra in $\Gamma^{\op}\Set$ functor $G$ of Example \ref{ex:async}.
\end{example}

\begin{example}[Dynamic networks]
\label{ex:dynamicnetworks}
The chromatic simplicial sets with decisions provide a simple way to extend our interpretation of patterns of communication to general dynamic network graphs. For this, we use as ``decision values'' the round number, an integral number (or even a history of communication in the general case). For any dynamic network graph $G$, giving a particular pattern of communication for each round $i \in \N$, we produce $F_G: \ \Psi \rightarrow \Psi^{\op}\Set$ as follows. For each $(A,D)$-simplex, with $D(a)=i$ for all $a \in A$, the common round number, we produce the chromatic simplicial set which is $F(A)$ for the $i$th pattern of communication in $G$, decorated with decision values, common to all agents, equal to $i+1$. 
\end{example}

\subsection{Relations between the models and interpreted protocols} 
\label{sec:interpretedprotocol}
\label{sec:rel}

We first note that there is an obvious forgetful functor 
$\nu: \ \Psi \rightarrow \Gamma$. From this forgetful functor, we can produce pairs of adjoint functors between $\Psi^{\op}\Set$ and $\Gamma^{\op}\Set$ as follows:


\begin{definition}
Consider the following functor:
$c_\nu: \ \Psi \rightarrow \Gamma^{\op}\Set$ defined:
\begin{itemize}
\item on objects by 
$c_\nu(A,D)=\Gamma[A]$)
\item on morphisms $f: (A,D)\rightarrow (B,E)$), we associate the morphism from $\Gamma[A]$ to $\Gamma[B]$ induced by the inclusion map $A \subseteq B$. 
\end{itemize}
We still call $c_\nu$ their left Yoneda extensions (as in Lemma \ref{fundamental}), from $\Psi^{\op}\Set$ to $\Gamma^{\op}\Set$. 
\end{definition}

By Lemma \ref{fundamental}, $c_\nu$ has a right adjoint we denote by $d_\nu$ from $\Gamma^{\op}\Set$ to $\Psi^{\op}\Set$.




A state-independent protocol is given by $F: \Psi^{\op}\Set \rightarrow \Psi^{\op}\Set$ that ``comes'' from a functor $\Gamma \rightarrow \Psi^{\op}\Set$, by composition with $d_\nu$. 




\begin{example}[Distributed robotics with unbounded visibility]
In the case of the multirobot system of Example \ref{ex:multirobot}, where we allow for unbounded visibility, it has been observed that task solvability reduces to standard task solvability for immediate snapshot models \cite{LCM}. Indeed, this is a consequence of the fact that in that case, the protocol complex model of Example \ref{ex:multirobot} is exactly the same as the one for the immediate snapshot model of Example \ref{ex:async}, or more precisely, that $c_\nu$ of the former one is isomorphic to the latter one. 
\end{example}

Let us now discuss particular state-dependent protocols which encode, for instance, the concrete look-compute-move protocols in distributed robotics (see Example \ref{ex:averageconsensus}): 

\begin{definition}[Concrete protocol] \label{def:concreteprotocol}
A concrete protocol is given by a functor $F: \ \Gamma^{\op}\Set \rightarrow \Gamma^{\op}\Set$ and local decision maps $f_a: \ \wp(\N) \times \mathcal{V}^\N \rightarrow \mathcal{V}$ which associates to each agent $a \in \N$ and each of the current values (for agents, including $a$), i.e. each set of participating agent $A \subseteq \N$ and a ``local'' value $v \in \mathcal{V}$, a ``local'' decision value for agent $a$, $f_a(A,v)$. 

Supposing that $F$ is agent preserving, that is, we can only have $B$ simplexes of $F(A)$ for $B \subseteq A$, then, we create a functor $F_f: \Psi^{\op}\Set \rightarrow \Psi^{\op}\Set$ from $F$ and $f$, as follows. It is defined by left Yoneda extension:  
\begin{itemize}
\item to $(B,E) \in \Psi$ we associate $F(\Gamma[B])$ that we ``decorate'' as follows so that it becomes an element of $\Psi^{\op}\Set$. Each $A$ simplex of $F(\Gamma[B])$ is 
decorated to become an $(A,D)$ simplex as follows. We let $D(a)=f_a(A,E)$ for all $a \in A$. This is correctly defined since $F$ is agent preserving
\item 
to a morphism $g: \ (B,E)\rightarrow (B',E')$
we associate $F(g)$ ($g$ being seen as a morphism of $\Gamma$, from $B$ to $B'$) which is defining a morphism of $\Psi^{\op}\Set$. 
\end{itemize}
\end{definition}

Note that any algebra $g: \ F(I) \rightarrow I$ lifts to give an algebra $g: \ F_f(I_d)\rightarrow I_d$ for any $I_d$ chromatic semi-simplicial set with decision values such that $c_\nu(I_d)=I$. 

\begin{example}[Averaging protocol]
\label{ex:averageconsensus}
We are now defining ``the asynchronous average consensus concrete protocol functor'' as a concrete protocol. Consider $G$, the immediate snapshot protocol complex of Example \ref{ex:async}, that we consider as a functor from  $\Gamma^{\op}\Set$ to $\Gamma^{\op}\Set$. Consider now $\CatValues$ to be $\R$, the real numbers, as local values for agents, and $f_a: \wp(\N)\times \mathcal{V}^\N \rightarrow \mathcal{V}$ for all $a \in \N$ agent, to be $$f_a(V,X)=X(a)+\alpha \sum\limits_{b\in V} (X(b)-X(a)),$$ for some $0<\alpha <1$. 
The asynchronous average consensus concrete protocol functor $F_A$ is the concrete protocol given by functor $G$ and these decisions maps $f_a$, according to Definition \ref{def:concreteprotocol}. 
\end{example}



\subsection{Semantics in free algebras over chromatic simplicial sets with decisions} 
\label{sec:semantics-cset-decision}

Recall the semantics in free algebras over csets from \Cref{sec:semantics-cset}. There are similar extensions of chromatic simplicial sets with decision values to simplicial models with decision values. We write $cset^d_\AP$ for the category of simplicial models with decision values.

In this paragraph, we suppose that $\Fun{F}$ is a functor from $\Psi^{op}\Set$ to $\Psi^{op}\Set$ obtained by Yoneda extension, and that we are given a $\Fun{F}$-algebra over some element $I \in cset^d_{AP}$. This is typically the case given a concrete protocol as in Definition \ref{def:concreteprotocol}. 

We extend the logics we will interpret with new predicates $q$ on a set $\AQ$ of predicates on tuples of elements of $\CatValues$, designed for giving properties of decision values of each agent at any step of a protocol. 
We interpret the extra bit in the logics as follows, denoting by $d_a$, $a \in \N$ the local decision value of agent $a$ in some world (that the context makes clear):



\begin{tabular}{llrl}
(Q) & $I^\infty,x \models q((d_a)_{a \in A})$ & iff & 
\mbox{$x$ is an $(A,D)$-simplex and $q((D(a))_{a \in A})$ holds} \\
\end{tabular}
\label{sec:semantics}


This extension of simplicial models to simplicial models for decisions is extremely powerful. Indeed, given a task specification as in Section \ref{sec:applications} and supposing that the number of simplexes in the input and in the output complexes is finite, we can encode each one of these by an integer, and any input-output specification $T$ as a predicate $q$ on values $\Values=\N$, enumerating the valid input-output pairs of worlds. Given that classical task specification uses any iterated protocol complex $P$, this amounts to using Equation (Diam) of the semantics of our temporal-epistemic logic of Section \ref{sec:semantics}. Task specification is then equivalent to $\Diamond q$. 

\subsection{Example: task specification of robot exploration} 


We now show that the language of temporal epistemic logic introduced in \Cref{sec:logics} can be used to specify tasks in distributed robotic systems, as also shown in \cite{robot-frames}, and that the free algebras model executions of the multi-robot system behavior captured in a protocol complex functor of chromatic simplicial sets with decisions from \Cref{sec:decision}.

\begin{example}[Exploration task]
Exploration is defined by a group of robots \(R_i\) at positions \(x_i\) capable of movement (updating their positions) and observing their environment. They execute look-compute-move cycles and their possible movement updates are defined by the robot's dynamics \footnote{The dynamics of a robot is defined using a system of ODEs that can be arbitrarily complex. As this is not the focus of this paper, we will consider only a point-mass system, i.e., a point in space capable of moving freely in any direction.} and encoded in their accessible states \(N_i(x_i)\), while the possible observations are constrained by an observation range \(r\). Formally, their environment \(E\) is the topological subspace of \(\R^d\) obtained from a basis of open balls using the robot's observation range.

Each robot decides a new position to take and how much of the space has been explored. This information is expressed as atomic propositions \(\xp_U \in \AQ\) associated to each possible observation \(U \in Open(E)\). Those propositions alongside the temporal epistemic logic of \Cref{sec:logics} form the exploration language \(\mathcal{L}_{exp}\) needed to describe the robot task of exploration.
$$ \varphi ::= \xp(U) \: \vert \: \neg \varphi \: \vert \: (\varphi \wedge \varphi) \: \vert \: K_\robot \varphi \: \vert \: D_B \varphi \: \vert \: \Diamond \varphi$$
where $\xp(U) \in \AQ$, $\robot \in A$ and $B \subseteq A$. 

A multi-robot system execution, modeled with an object \(I^{\infty}\) of the free \(\Fun{F}_R\)-algebra over some ``set'' of initial states \(I \in cset_\AQ\), is said to be consistent with exploration\cite{robot-frames} if it satisfies:

\begin{itemize}
  \item Exploration agency: an explored region has been explored by at least one robot.
    \begin{equation}
    \label{eq:agency}
      I^{\infty} \models \xp(U) \rightarrow  D_{\Pi} \xp(U)
    \end{equation}

  \item Exploration independence: distributed exploration knowledge comes from the union of the respective robot's explored regions.
    \begin{equation}
    \label{eq:independence}    
      I^{\infty} \models D_{B} \xp(U) \rightarrow \bigwedge_{\robot \in B}K_\robot  exp(V_\robot) \textrm{ such that } \bigcup_{\robot \in B} V_\robot = U
    \end{equation}

  \item Stable exploration statements: an explored region remains explored forever.\\
  For any region $U$
    \begin{equation}
    \label{eq:stable_prop}
      I^{\infty} \models \bigwedge_{\xp(U) \in \mathcal{L}_{exp}^-} \xp(U) \rightarrow \Box \xp(U)
    \end{equation}

  \item Stable exploration knowledge: the knowledge of an explored region remains known forever.\\
  For any $\robot \in \Pi$, and any region $U$
      \begin{equation}
      \label{eq:perfectrecall}
          I^{\infty} \models K_{\robot} \xp(U) \rightarrow \Box K_\robot \xp(U)
      \end{equation}

\end{itemize}

It becomes evident that exploration of a space \(X\) is satisfied whenever \(I^{\infty} \models \Diamond D_{A}\xp(X)\), when a robot terminates as soon as it knows that the entire space has been explored. Alternatively, one can require that robots only terminate whenever it is common knowledge that the mission has finished, i.e., \(I^{\infty} \models \Diamond \bigwedge_{\robot \in A} K_\robot X\).
\end{example}


\section{Conclusion} 

\label{sec:conclusion}

In this paper, we showed that the round-based protocol complex approach of \cite{herlihy} falls nicely into the categorical framework of algebras over chromatic augmented pre-simplicial sets (csets). The protocol complexes themselves are naturally defined as functors from csets to csets, that are Yoneda extensions of very simple functors. This algebraic formulation naturally leads to a temporal-epistemic logics, that is useful for specification, and verification purposes. 

Finally, we paved the way for future extensions, leading to wider application areas such as distributed robotics, where a notion of local state has to be taken into account for even defining the corresponding protocol complexes. 

As a note in Appendix \ref{app:topologicalinvariants} (this is not the main subject of this paper), this formulation allows simple reasoning for finding topological invariants of the protocol complex. This is indeed linked to some of the logical properties that we develop in this paper, as developed in e.g. \cite{gandalf}, but this is left for future work. 


\begin{thebibliography}{10}
\providecommand{\url}[1]{\texttt{#1}}
\providecommand{\urlprefix}{URL }
\expandafter\ifx\csname urlstyle\endcsname\relax
  \providecommand{\doi}[1]{doi:\discretionary{}{}{}#1}\else
  \providecommand{\doi}{doi:\discretionary{}{}{}\begingroup
  \urlstyle{rm}\Url}\fi
\providecommand{\eprint}[2][]{\url{#2}}

\bibitem{herlihy}
Herlihy M, Feichtner-Kozlov D, Rajsbaum S.
\newblock Distributed Computing Through Combinatorial Topology.
\newblock Morgan Kaufmann, Amsterdam, 2014.
\newblock ISBN 978-0-12-404578-1.

\bibitem{GoubaultLR21simplicial}
Goubault {\'{E}}, Ledent J, Rajsbaum S.
\newblock A Simplicial Complex Model for Dynamic Epistemic Logic to study
  Distributed Task Computability.
\newblock \emph{Information and Computation}, 2021.
\newblock \textbf{278}:104597.
\newblock \doi{10.1016/j.ic.2020.104597}.

\bibitem{LTL}
Manna Z, Pnueli A.
\newblock Temporal verification of reactive systems: safety.
\newblock Springer-Verlag, Berlin, Heidelberg, 1995.
\newblock ISBN 0387944591.

\bibitem{linearizability}
Herlihy MP, Wing JM.
\newblock Linearizability: a correctness condition for concurrent objects.
\newblock \emph{ACM Trans. Program. Lang. Syst.}, 1990.
\newblock \textbf{12}(3):463–492.
\newblock \doi{10.1145/78969.78972}.
\newblock \urlprefix\url{https://doi.org/10.1145/78969.78972}.

\bibitem{MoranW87}
Moran S, Wolfstahl Y.
\newblock Extended impossibility results for asynchronous complete networks.
\newblock \emph{Information Processing Letters}, 1987.
\newblock \textbf{26}(3):145--151.
\newblock \doi{https://doi.org/10.1016/0020-0190(87)90052-4}.
\newblock
  \urlprefix\url{https://www.sciencedirect.com/science/article/pii/0020019087900524}.

\bibitem{BiranMZ90}
Biran O, Moran S, Zaks S.
\newblock A Combinatorial Characterization of the Distributed 1-Solvable Tasks.
\newblock \emph{J. Algorithms}, 1990.
\newblock \textbf{11}(3):420--440.
\newblock \doi{10.1016/0196-6774(90)90020-F}.
\newblock \urlprefix\url{https://doi.org/10.1016/0196-6774(90)90020-F}.

\bibitem{unifying}
Casta{\~{n}}eda A, Rajsbaum S, Raynal M.
\newblock Unifying Concurrent Objects and Distributed Tasks:
  Interval-Linearizability.
\newblock \emph{J. {ACM}}, 2018.
\newblock \textbf{65}(6):45:1--45:42.
\newblock \doi{10.1145/3266457}.

\bibitem{opodis18}
Goubault {\'{E}}, Ledent J, Mimram S.
\newblock Concurrent Specifications Beyond Linearizability.
\newblock In: 22nd International Conference on Principles of Distributed
  Systems, {OPODIS} 2018, December 17-19, 2018, Hong Kong, China. 2018 pp.
  28:1--28:16.
\newblock \doi{10.4230/LIPICS.OPODIS.2018.28}.

\bibitem{FHMVbook}
R~Fagin YM J~Halpern, Vardi M.
\newblock Reasoning About Knowledge.
\newblock MIT Press, 1995.

\bibitem{Knight}
Knight S.
\newblock The Epistemic View of Concurrency Theory. (Le point de vue
  epist{\'{e}}mique de th{\'{e}}orie de la concurrence).
\newblock Ph.D. thesis, {\'{E}}cole Polytechnique, Palaiseau, France, 2013.

\bibitem{infcomp}
Goubault {\'{E}}, Ledent J, Rajsbaum S.
\newblock A simplicial complex model for dynamic epistemic logic to study
  distributed task computability.
\newblock \emph{Inf. Comput.}, 2021.
\newblock \textbf{278}:104597.
\newblock \doi{10.1016/J.IC.2020.104597}.
\newblock \urlprefix\url{https://doi.org/10.1016/j.ic.2020.104597}.

\bibitem{boletin}
Goubault E, Kniazev R, Ledent J, Rajsbaum S.
\newblock Simplicial models for the epistemic logic of faulty agents.
\newblock \emph{Bol. Soc. Mat. Mex.}, 2024.
\newblock \textbf{30}(90).
\newblock \urlprefix\url{https://doi.org/10.1007/s40590-024-00656-x}.

\bibitem{LICS2023}
Goubault {\'{E}}, Kniazev R, Ledent J, Rajsbaum S.
\newblock Semi-Simplicial Set Models for Distributed Knowledge.
\newblock In: {LICS}. 2023 pp. 1--13.
\newblock \doi{10.1109/LICS56636.2023.10175737}.
\newblock \urlprefix\url{https://doi.org/10.1109/LICS56636.2023.10175737}.

\bibitem{gandalf}
Goubault {\'{E}}, Ledent J, Rajsbaum S.
\newblock A Simplicial Complex Model for Dynamic Epistemic Logic to study
  Distributed Task Computability.
\newblock In: Orlandini A, Zimmermann M (eds.), Proceedings Ninth International
  Symposium on Games, Automata, Logics, and Formal Verification, GandALF 2018,
  Saarbr{\"{u}}cken, Germany, 26-28th September 2018, volume 277 of
  \emph{{EPTCS}}. 2018 pp. 73--87.
\newblock \doi{10.4204/EPTCS.277.6}.

\bibitem{dynamicnetworks}
Winkler K, Schmid U.
\newblock An Overview of Recent Results for Consensus in Directed Dynamic
  Networks.
\newblock \emph{Bull. {EATCS}}, 2019.
\newblock \textbf{128}.
\newblock
  \urlprefix\url{http://bulletin.eatcs.org/index.php/beatcs/article/view/581/585}.

\bibitem{SMLCategories}
MacLane S.
\newblock Categories for the Working Mathematician.
\newblock Springer-Verlag, New York, 1971.
\newblock Graduate Texts in Mathematics, Vol. 5.

\bibitem{Varieties}
Ad{\'{a}}mek J, Porst H.
\newblock On Varieties and Covarieties in a Category.
\newblock \emph{Math. Struct. Comput. Sci.}, 2003.
\newblock \textbf{13}(2):201--232.

\bibitem{trnkova1975free}
Trnkov{\'a} V, Ad{\'a}mek J, Koubek V, Reiterman J.
\newblock Free algebras, input processes and free monads.
\newblock \emph{Commentationes Mathematicae Universitatis Carolinae}, 1975.
\newblock \textbf{16}(2):339--351.

\bibitem{MosesR2002}
Moses Y, Rajsbaum S.
\newblock A Layered Analysis of Consensus.
\newblock \emph{SIAM J. Comput.}, 2002.
\newblock \textbf{31}(4):989--1021.

\bibitem{sep-dynamic-epistemic}
Baltag A, Renne B.
\newblock Dynamic Epistemic Logic.
\newblock In: The Stanford Encyclopedia of Philosophy, see
  {\url{https://plato.stanford.edu/archives/win2016/entries/dynamic-epistemic/}}.
  Metaphysics Research Lab, Stanford University, 2016.

\bibitem{DEL:2007}
Ditmarsch Hv, van~der Hoek W, Kooi B.
\newblock Dynamic Epistemic Logic.
\newblock Springer, 2007.
\newblock \doi{10.1007/978-1-4020-5839-4}.

\bibitem{baltagMS:98}
Baltag A, Moss L, Solecki S.
\newblock The logic of common knowledge, public announcements, and private
  suspicions.
\newblock In: TARK VII. 1998 pp. 43--56.
\newblock \doi{10.1007/978-3-319-20451-2_38}.

\bibitem{ss}
Eilenberg S, Zilber JA.
\newblock Semi-simplicial complexes and singular homology.
\newblock \emph{Annals of Mathematics}, 1950.
\newblock pp. 499--513.

\bibitem{Riehl1}
Riehl E.
\newblock A Leisurely Introduction to Simplicial Sets.

\bibitem{topos}
Mac~Lane S, Moerdijk I.
\newblock Sheaves in Geometry and Logic a First Introduction to Topos Theory.
\newblock Springer New York, 1992.
\newblock ISBN 9781461209270 1461209277.

\bibitem{PGMZCalculus}
Gabriel P, Zisman M.
\newblock Calculus of fractions and homotopy theory.
\newblock Springer-Verlag Berlin, New York, 1967.

\bibitem{BorowskyG92}
Borowsky E, Gafni E.
\newblock Immediate Atomic Snapshots and Fast Renaming (Extended Abstract).
\newblock In: Anderson J, Toueg S (eds.), Proceedings of the Twelth Annual
  {ACM} Symposium on Principles of Distributed Computing, Ithaca, New York,
  USA, August 15-18, 1993. {ACM}, 1993 pp. 41--51.
\newblock \doi{10.1145/164051.164056}.
\newblock \urlprefix\url{https://doi.org/10.1145/164051.164056}.

\bibitem{updatescan1}
Goubault {\'{E}}, Mimram S, Tasson C.
\newblock Geometric and combinatorial views on asynchronous computability.
\newblock \emph{Distributed Comput.}, 2018.
\newblock \textbf{31}(4):289--316.

\bibitem{updatescan2}
Kozlov DN.
\newblock {Chromatic subdivision of a simplicial complex}.
\newblock \emph{Homology, Homotopy and Applications}, 2012.
\newblock \textbf{14}(2):197 -- 209.

\bibitem{Kuhn11dynamic}
Kuhn F, Oshman R.
\newblock Dynamic Networks: Models and Algorithms.
\newblock \emph{SIGACT News}, 2011.
\newblock \textbf{42}(1):82–96.
\newblock \doi{10.1145/1959045.1959064}.
\newblock \urlprefix\url{https://doi.org/10.1145/1959045.1959064}.

\bibitem{patternmodels}
Castañeda A, van Ditmarsch H, Rosenblueth DA, Velázquez DA.
\newblock Pattern Models: A Dynamic Epistemic Logic For Distributed Systems.
\newblock \emph{The Computer Journal}, 2024.
\newblock \textbf{67}(7):2421--2440.
\newblock \doi{10.1093/comjnl/bxae016}.
\newblock \urlprefix\url{https://doi.org/10.1093/comjnl/bxae016}.

\bibitem{CouloumaGP15}
Coulouma {\'{E}}, Godard E, Peters JG.
\newblock A characterization of oblivious message adversaries for which
  Consensus is solvable.
\newblock \emph{Theor. Comput. Sci.}, 2015.
\newblock \textbf{584}:80--90.
\newblock \doi{10.1016/J.TCS.2015.01.024}.
\newblock \urlprefix\url{https://doi.org/10.1016/j.tcs.2015.01.024}.

\bibitem{kurz}
Kurz A.
\newblock Logics for Coalgebras and Applications to Computer Science.
\newblock BoD - Books on Demand, 2001.
\newblock ISBN 9783831112494.

\bibitem{maclane1992sheaves}
Mac~Lane S, Moerdijk I.
\newblock Sheaves in Geometry and Logic a First Introduction to Topos Theory,
  1992.
\newblock
  \urlprefix\url{http://link.springer.com/book/10.1007/978-1-4612-0927-0}.

\bibitem{Karazeris}
Karazeris P, Protsonis G.
\newblock Left Kan Extensions Preserving Finite Products.
\newblock
  \urlprefix\url{http://old.math.upatras.gr/~pkarazer/publications/topsift.pdf}.

\bibitem{STACS22}
Goubault {\'{E}}, Ledent J, Rajsbaum S.
\newblock A Simplicial Model for KB4{\_}n: Epistemic Logic with Agents That May
  Die.
\newblock In: 39th International Symposium on Theoretical Aspects of Computer
  Science, {STACS} 2022. 2022 p. 33:1–33:20.
\newblock \doi{10.4230/LIPIcs.STACS.2022.33}.

\bibitem{arxivfaulty}
Goubault E, Kniazev R, Ledent J, Rajsbaum S.
\newblock Simplicial Models for the Epistemic Logic of Faulty Agents.
\newblock \emph{CoRR}, 2023.
\newblock \textbf{abs/2311.01351}.
\newblock \doi{10.48550/ARXIV.2311.01351}.
\newblock \eprint{2311.01351},
  \urlprefix\url{https://doi.org/10.48550/arXiv.2311.01351}.

\bibitem{LCM}
Alc\'{a}ntara M, Casta\~{n}eda A, Flores-Pe\~{n}aloza D, Rajsbaum S.
\newblock The topology of look-compute-move robot wait-free algorithms with
  hard termination.
\newblock \emph{Distrib. Comput.}, 2019.
\newblock \textbf{32}(3):235–255.
\newblock \doi{10.1007/s00446-018-0345-3}.
\newblock \urlprefix\url{https://doi.org/10.1007/s00446-018-0345-3}.

\bibitem{robot-frames}
Cignarale G, Felber S, Goubault E, Hummes~Flores B, Rincon~Galeana H.
\newblock Knowledge in Multi-Robot Systems: An Interplay of Dynamics,
  Computation and Communication.
\newblock \doi{10.48550/arXiv.2501.18309}.
\newblock \eprint{2501.18309}.

\bibitem{Bjorner}
Björner A.
\newblock Nerves, fibers and homotopy groups.
\newblock \emph{Journal of Combinatorial Theory, Series A}, 2003.
\newblock \textbf{102}(1):88--93.
\newblock \doi{https://doi.org/10.1016/S0097-3165(03)00015-3}.
\newblock
  \urlprefix\url{https://www.sciencedirect.com/science/article/pii/S0097316503000153}.

\bibitem{collapsible}
Éric Goubault, Mimram S, Tasson C.
\newblock Iterated Chromatic Subdivisions are Collapsible.
\newblock \emph{Applied Categorical Structures}, 2014.
\newblock p. 1–42.
\newblock \doi{10.1007/s10485-014-9383-6}.
\newblock \urlprefix\url{http://dx.doi.org/10.1007/s10485-014-9383-6}.

\end{thebibliography}



\newpage

\appendix

\section{Topological invariants of the protocol complex}

\label{app:topologicalinvariants}
The formulation of the protocol complex as a certain left Yoneda extension allows, in some situations, to get the homology type of the protocol complex, for free. As well-known, this is instrumental for showing impossibility of some fault-tolerant decision tasks, such as consensus, in e.g. asynchronous distributed systems. 

First, we remind the reader of the generalized nerve theorem of \cite{Bjorner}. 
But before that, we need a few definitions:

\begin{definition}
Let $(X_i)_{i\in I}$ be a family of sets. The nerve of this family of sets is the simplicial complex $\mathcal{N}(X_i)$ which has:
\begin{itemize}
\item as vertices, the indexes $i \in I$, 
\item as simplexes $\sigma$, subsets of $I$ such that $\bigcap\limits_{i\in \sigma} X_i \neq \emptyset$.
\end{itemize}
\end{definition}

Said differently, the nerve $\mathcal{N}(X_i)$ is the nerve (in the category theory sense, or in the classical poset meaning) of the diagram $D^X$ given by the canonical inclusion maps of all non-empty intersections
$X_{{i_1}}\cap \cdots \cap X_{{i_t}}$ into all its sub-intersections.

\begin{theorem}[\cite{Bjorner}]
\label{thm:nerve}
Let $\Delta$ be a connected regular CW complex and $(\Delta_i)_{i\in I}$ a family of subcomplexes such that $\Delta=\bigcup\limits_{i\in I} \Delta_i$ (i.e. a cover of $\Delta$). Suppose that every non-empty finite intersection $\Delta_{i_1} \cap \Delta_{i_2} \cap \cdots \cap \Delta_{i_t}$ is $(k-t+1)$-connected, $t\geq 1$. Then there is a map $f$ from the $\Delta$ to the geometric realization of the nerve $\mathcal{N}(\Delta_i)$ inducing isomorphisms of homotopy groups $f^*_j: \ \pi_j(\Delta)\sim \pi_j(\mathcal{N}(\Delta_j))$ for all $j\leq k$.
\end{theorem}


\begin{example}[Immediate snapshot example]
By \cite{collapsible}, the image of any simplex by $G_!$ is collapsible, hence contractible. A direct application of Theorem \ref{thm:nerve} shows that $G_!(I)$ has the same homotopy type as $I$, for any input simplicial {\em complex} $I$. 

Indeed, for any input complex $I$, $I$ is covered by the subcomplexes $\Delta^I_\sigma$ generated by all maximal simplexes $\sigma$ of $I$. All these complexes are contractible, with non-empty intersections being a subcomplex generated by a unique simplex, hence being contractible as well. This is due to the fact that $I$ is a simplicial complex, it would not be true in general if $I$ was a simplicial set (the intersection of two closed simplexes may be disconnected, for instance). Furthermore, $I$ is the colimit of the $\Delta^I_\sigma$ over the diagram $D^I$ given by the canonical inclusion maps of all non-empty intersections
$\Delta^I_{\sigma_{i_1}}\cap \cdots \cap \Delta^I_{\sigma_{i_t}}$, identified with the subcomplex $\Delta^I_{i_1,\cdots,i_t}$ generated by a unique simplex, into all its sub-intersections $\Delta^I_{i_1,\cdots,\hat{i}_k,\cdots,i_t}$. By Theorem \ref{thm:nerve} applied to a large enough $k$ (at least the maximal dimension of simplexes in $I$), $\pi_j(I)$ is isomorphic to $\pi_j(\mathcal{N}(\Delta^I_i))$, i.e. $I$ has the same homotopy type as $\mathcal{N}(\Delta^I_i)$.

As $G_!$ commutes with colimits, Lemma \ref{lem:preserve}, $G_!(I)$ is the colimit over $G_!(\Delta^I_{\sigma})$. $\sigma$ being a maximal simplex, this is the representable functor in $\Gamma^{\op} \Set$ $\Gamma(.,A_\sigma)$, where $A_\sigma$ is an object of $\Gamma$ (a set of colors). By definition of a left Yoneda extension, $G_!(\Delta^I_{\sigma})$ is thus equal to the cset $G(A_\sigma)$. Furthermore, $G_!(D^I)$ is the diagram of inclusions $G_!(\Delta^I_{i_1,\cdots,i_t})$ into $G_!(\Delta^I_{i_1,\cdots,\hat{i}_k,\cdots,i_t})$, by Lemma \ref{lem:preserve} again, which again, as these complexes $\Delta^I_{i_1,\cdots,i_t}$ are generated by a unique simplex $\Gamma(.,A_{i_1,\cdots,i_t})$, is the diagram of inclusions $G(A_{i_1,\cdots,i_t})$ into $G(A_{i_1,\cdots,\hat{i}_k,\cdots,i_t})$. This means that the $G(A_\sigma)$ where $\sigma$ ranges over the maximal simplexes of $I$ form a cover of $G_!(I)$, and has the same underlying diagram of inclusions as $D^I$. By Theorem \ref{thm:nerve}, as $G(A_\sigma)$ and all its non-empty intersections $G(A_{i_1,\cdots,i_t})$ are contractible, $G_!(I)$ has the same homotopy type again as $\mathcal{N}(D^I)$ i.e. as $I$. 
\end{example}

\begin{example}[Synchronous broadcast example]
Unfortunately, even the generalized nerve theorem of \cite{Bjorner} does not help with the synchronous broadcast case. In this case, we need to use the Mayer-Vietoris exact sequence: 

\begin{center}
    \begin{tikzcd}[arrow style=math font,cells={nodes={text height=2ex,text depth=0.75ex}}]
    & & 0 \arrow[draw=none]{d}[name=X,shape=coordinate]{} \\
       H_{0}[X] \arrow[curarrow=X]{urr}{} 
       &  H_0[X_1]\oplus H_0[X_2]    \arrow[l] \arrow[draw=none]{d}[name=Y, shape=coordinate]{} & \arrow[l] \cdots \\
       H_{i}[X] \arrow[curarrow=Y]{urr}{} & H_{i}[X_1]\oplus H_i[X_2] \arrow[l] \arrow[draw=none]{d}[name=Z,shape=coordinate]{} &  H_i[X_1\cap X_2]) \arrow[l] \\
       H_{i+1}[X] \arrow[curarrow=Z]{urr}{} & H_{i+1}[X_1]\oplus H_{i+1}[X_2] \arrow[l] & \cdots \arrow[l]
   \end{tikzcd}
\end{center}
This allows for instance, for an inductive computation of the homology groups of $F_!(X)$ where $X$ is a binary sphere of dimension $n\geq 2$. In that case, we decompose $X$ as $X_1 \cup X_2$ where $X_1$ is the north hemisphere of $X$, $X_2$ is the south hemisphere of $X$ and $X_1\cap X_2$ is the binary sphere of dimension $n-1$. As $F_!$ preserves colimits, Lemma \ref{lem:preserve}, we get:
\begin{center}
    \begin{tikzcd}[arrow style=math font,cells={nodes={text height=2ex,text depth=0.75ex}}]
    & & 0 \arrow[draw=none]{d}[name=X,shape=coordinate]{} \\
       H_{0}[F_!(X)] \arrow[curarrow=X]{urr}{} 
       &  H_0[F_!(X_1)]\oplus HM_1[F_!(X_2)]    \arrow[l] \arrow[draw=none]{d}[name=Y, shape=coordinate]{} & \arrow[l] \cdots \\
       H_{i}[F_!(X)] \arrow[curarrow=Y]{urr}{} & H_{i}[F_!(X_1)]\oplus H_i[F_!(X_2)] \arrow[l] \arrow[draw=none]{d}[name=Z,shape=coordinate]{} &  H_i[F_!(X_1\cap X_2)]) \arrow[l] \\
       H_{i+1}[F_!(X)] \arrow[curarrow=Z]{urr}{} & H_{i+1}[F_!(X_1)]\oplus H_{i+1}[F_!(X_2)] \arrow[l] & \cdots \arrow[l]
   \end{tikzcd}
\end{center}
\noindent where the diagram is uniformly 0 from dimension $n+1$ on. This allows in particular to get simple inequalities on the rank of the various $H_i[F_!(X)]$. 
%
\end{example}

\end{document}